\documentclass[aps,pra,reprint,superscriptaddress]{revtex4-2}
\usepackage{amsmath}
\usepackage{amssymb}
\usepackage{amsthm}
\usepackage[bookmarks=false]{hyperref}
\usepackage{bm,bbm}
\hypersetup{colorlinks=true,citecolor=blue,linkcolor=red,%
urlcolor=blue,pdfstartview=FitH,bookmarksopen=true}

\usepackage[T1]{fontenc}
\usepackage{mathpazo}
\usepackage{dsfont}
\usepackage{braket}

\usepackage{diagbox}

\usepackage{graphicx}

\usepackage{caption}
\DeclareCaptionJustification{myjust}{\fulljustify}
\makeatletter
\providecommand\fulljustify{%
  \let\\\@centercr
  \leftskip\z@%
  \rightskip\z@%
  \parfillskip\z@\@plus 1fill\relax%
}
\makeatother

\usepackage[justification=myjust,format=plain]{caption}
\usepackage[justification=myjust,format=plain]{subcaption}
\usepackage{url}
\usepackage{hhline}

\newtheorem{definition}{Definition}
\newtheorem{lemma}[definition]{Lemma}
\newtheorem{theorem}[definition]{Theorem}

\newtheorem{proposition}[definition]{Proposition}
\newtheorem{conjecture}[definition]{Conjecture}

\global\long\def\one{\mathds{1}}
\global\long\def\trace{\operatorname{Tr}}
\global\long\def\const{\operatorname{const}}
\global\long\def\ketbra#1#2{\ket{#1}\!\bra{#2}}

\newcommand{\dR}{\mathds{R}}
\newcommand{\maxover}[1][]{\underset{#1}{\mathrm{max}}}

\newcommand{\optover}[1][]{\underset{#1}{\mathrm{opt}}}
\newcommand{\subto}{\mathrm{~s.t.}}

\DeclareMathOperator{\vol}{vol}

\DeclareMathOperator{\SEP}{SEP}
\newcommand{\Sep}{\mathrm{SEP}}

\newcommand{\norm}[1]{\left\lVert#1\right\rVert}

\usepackage[normalem]{ulem}
\usepackage{color}

\begin{document}

\title{Confident entanglement detection via separable numerical range}

\author{Timo Simnacher}
\affiliation{Naturwissenschaftlich-Technische Fakult\"at, Universit\"at Siegen,
Walter-Flex-Str. 3, D-57068 Siegen, Germany}
\affiliation{Faculty of Physics, Astronomy and Applied Computer Science,
Jagiellonian University, ul. {\L}ojasiewicza 11, 30-348 Krak{\'o}w, Poland}

\author{Jakub Czartowski}
\affiliation{Faculty of Physics, Astronomy and Applied Computer Science,
Jagiellonian University, ul. {\L}ojasiewicza 11, 30-348 Krak{\'o}w, Poland}

\author{Konrad Szyma{\'n}ski}
\affiliation{Faculty of Physics, Astronomy and Applied Computer Science,
Jagiellonian University, ul. {\L}ojasiewicza 11, 30-348 Krak{\'o}w, Poland}

\author{Karol {\.Z}yczkowski}
\affiliation{Faculty of Physics, Astronomy and Applied Computer Science,
Jagiellonian University, ul. {\L}ojasiewicza 11, 30-348 Krak{\'o}w, Poland}
\affiliation{Center for Theoretical Physics, Polish Academy of Sciences,
al. Lotnik{\'o}w 32/46, 02-668 Warszawa, Poland}

\date{\today}

\begin{abstract}
  We investigate the joint (separable) numerical range of multiple measurements, i.e., the regions of expectation values accessible with (separable) quantum states for given observables. This not only enables efficient entanglement detection, but also sheds light on the geometry of the set of quantum states.
  More precisely, in an experiment, if the confidence region for the obtained data and the separable numerical range are disjoint, entanglement is reliably detected. Generically, the success of such an experiment is more likely the smaller the separable numerical range is compared to the standard numerical range of the observables measured. We quantify this relation using the ratio between these two volumes and show that it cannot be arbitrarily small, giving analytical bounds for any number of particles, local dimensions as well as number of measurements. Moreover, we explicitly compute the volume of separable and standard numerical range for two locally traceless two-qubit product observables, which are of particular interest as they are easier to measure in practice. Furthermore, we consider typical volume ratios for generic observables and extreme instances.
\end{abstract}

\maketitle

\section{Introduction}
Quantum physics allows for correlations that are impossible in classical physics.
One of the most fundamental and remarkable manifestation of genuine quantum correlations is entanglement \cite{entanglement}, which is required for quantum steering \cite{steering} and nonlocality \cite{nonlocality}.
A bipartite quantum state $\rho_{AB}$ is called \textit{separable} if it can be written as
\begin{equation}
  \rho_{AB} = \sum_j p_j \rho_A^{(j)} \otimes \rho_B^{(j)},
\end{equation}
where the $p_j$ form a probability distribution and $\rho_A^{(j)}$, $\rho_B^{(j)}$ are quantum states of the first and second particle, respectively.
If the state $\rho_{AB}$ cannot be written in this form, then it is called \textit{entangled}.
In quantum information theory, entanglement is considered a resource that can be used to achieve certain tasks, e.g., entanglement has proved useful for quantum metrology, communication, and computation.
Hence, verifying entanglement in experiments is essential and many methods have been developed \cite{entanglementdetection}.

Here, we consider multiple observables and the measurements of their expectation values.
The accessible regions for vectors of expectation values of $k$ observables for (separable) quantum states is given by the (separable) numerical range.
If, for a given state, the measurement results give a point outside the separable numerical range, then that state must have been entangled.
Thus, we first provide insight into how this approach can be used for entanglement detection.
Secondly, since the (separable) numerical range is ultimately an affine transformation of a projection of the (separable) quantum state space, our investigation also sheds light on the geometry of quantum states, especially on the relation between the separable and the general quantum state space geometry.
Finally, in practical experiments, statistical and systematic errors lead to a confidence region instead of a single point contained in the numerical range of the observables measured.
We compare the volumes of separable and standard numerical range to gain intuition on how useful the considered measurements are for entanglement detection in practical scenarios.

\subsection{Methods}
\begin{figure}[t!]
  \centering
  \includegraphics[width=0.7\linewidth]{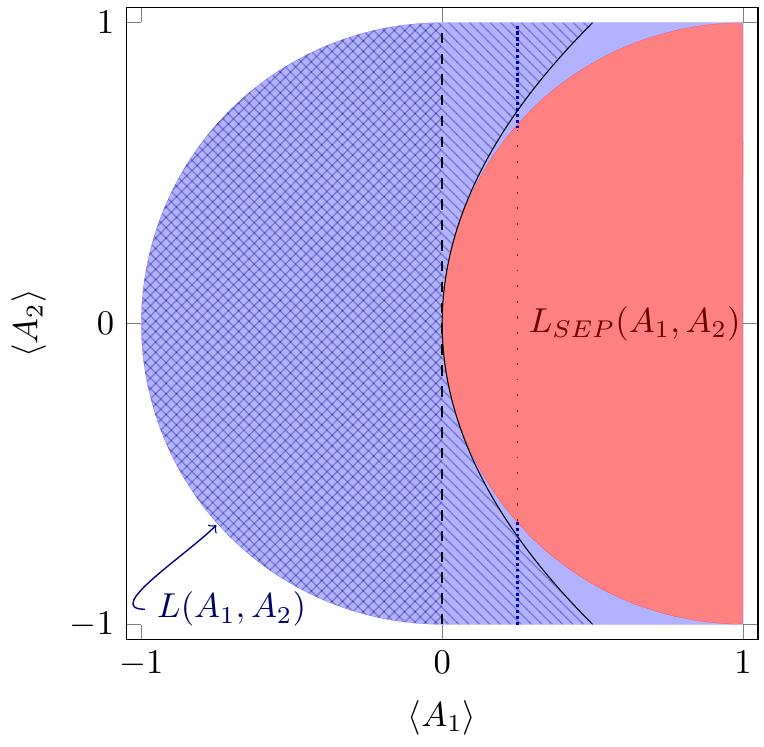}
  \caption{An illustration of different variants of entanglement witnesses.
           $L(A_1,A_2)$ and $L_{\SEP}(A_1,A_2)$ are the (separable) numerical range of the two-qubit observables $A_1 = \ket{00}\bra{11} + \ket{01}\bra{01} + \ket{10}\bra{10} + \ket{11}\bra{00}$ and $A_2 = -i \ket{00}\bra{11} + \ket{01}\bra{01} - \ket{10}\bra{10} + i \ket{11}\ket{00}$.
           The dashed line indicates the entanglement witness $W = A_1$, the dotted line the ultrafine entanglement witnesses with $\langle W_1 \rangle = \langle A_1 \rangle = 1/4 = \omega_1$ and $W_2 = \sqrt{7}/4 - A_2$ or $W_2 = (1-\sqrt{7})/4 + A_2$, and the curved solid line the nonlinear entanglement witness $\langle W_{\text{NL}} \rangle = \langle A_1 \rangle - {\langle A_2 \rangle}^2$.}
  \label{fig:compwit}
\end{figure}
For any entangled state $\rho$, there exists an observable $W$ such that $\trace W\rho < 0$ and $\trace W\sigma \ge 0$ for all separable states $\sigma$.
An observable with these properties is called an entanglement witness \cite{entanglementdetection}.
Entanglement witnesses are the standard tool for entanglement detection and are employed ubiquitously. 
Although a single measurement, that is repeated sufficiently many times, suffices to detect entanglement, this measurement, i.e. the entanglement witness, might be a highly entangled observable and therefore hard to implement in practice.
Ultrafine entanglement witnesses extend the concept of entanglement witnesses by taking into account multiple measurements for more reliable detection \cite{ultrafine,marinikolai}.
Also, measurements that are easier to implement, such as product observables, can be combined to simplify the detection in the experiment.
To do so, a first observable $W_1$ determines via measurement a subset of states that need to be considered constrained by $\trace W_1\rho = \omega_1$, where $\omega_1$ is the obtained measurement result.
Then, the second observable $W_2$ only needs to satisfy that $\trace W\sigma \ge 0$ for all separable states $\sigma$ with $\trace W_1\sigma = \omega_1$ allowing for more effective entanglement detection.
Lastly, nonlinear entanglement witnesses combine multiple measurements in a nonlinear fashion to improve entanglement detection \cite{nonlin1,nonlin2,nonlin3,nonlin4}.
While ordinary and ultrafine entanglement witnesses only provide a polyhedral approximation to the separable numerical range, nonlinear witnesses can observe the structure of the convex set generally better; see Fig.~\ref{fig:compwit} for a schematic visualization of the different methods.

In contrast to the different variants of entanglement witnesses, for the general situation of multiple measurements, the joint (separable) numerical range provides a comprehensive framework to tackle the problem of entanglement detection \cite{numericalrangeweb,gawron2010,wutang,czartowski2019}.
Throughout the paper, we are going to use the following notation.
\begin{definition}
  Let $A_1,\dots,A_k$ be Hermitian operators.
  Then, the set
  \begin{equation}
    L_X(A_1,\dots,A_k) = \{ (\trace\rho A_1, \dots, \trace\rho A_k) \in \dR^k \,|\, \rho \in X \} 
  \end{equation}
  is called the \textit{joint (restricted) numerical range} of $A_1,\dots,A_k$ where the set $X$ restricts the accessible states.
  If $X$ is the set of all quantum states, $L$ is simply the (joint) numerical range of $A_1,\dots,A_k$;
  if it is the set of all separable quantum states, $L_\Sep$ is called the separable (joint) numerical range of $A_1,\dots,A_k$.
\end{definition}
Furthermore, the Euclidean volume of $L_X(A_1,\dots,A_k)$ is denoted by $\vol L_X$.
This is the natural volume measure for experiments since it is the relevant measure for the confidence region obtained by measuring a given quantum state.

\subsection{Experimental confidence region}\label{sec:expconf}
An entangled state whose entanglement is in principle detectable by observables $A_1,\dots,A_k$ corresponds to a point in the numerical range $L(A_1,\dots,A_k)$ which is outside the separable numerical range $L_\Sep(A_1,\dots,A_k)$.
In a realistic experiment, however, only a finite number of measurement results can be collected which leads to an $\alpha$-confidence region inside the numerical range that covers the exact point given by the underlying state with probability at least $1-\alpha$.
Such a confidence region can be obtained, e.g., using Hoeffding's tail inequality \cite{hoeffding}.
In this case, it holds that the probability of an estimator $a_j$ for the expectation value of $A_j$ having at least a distance of $t$ from the actual expectation value is bounded by
\begin{equation}
  P\left(\Delta_j \ge t \right) \le 2 \exp\left\{ - \frac{2m t^2}{\left[\lambda_{\max}(A_j) - \lambda_{\min}(A_j)\right]^2} \right\},
\end{equation}
where $\Delta_j = \left| a_j - \trace \rho A_j \right|$ and the estimator $a_j$ is obtained as the average from $m$ measurements on the state $\rho$. Further, $\lambda_{\max}(A_j)$, $\lambda_{\min}(A_j)$ are the largest and smallest eigenvalues of $A_j$, respectively.
Since the individual measurements are independent, so are the estimators $a_j$.
Let us rescale the observables such that $\lambda_{\max}(A_j) - \lambda_{\min}(A_j) = 1$.
Then, we obtain a confidence region in the form of a hyperrectangle via
\begin{equation}
  P\left(\exists j:\Delta_j \ge t_j \right) \le 1 - \left[ 1 - 2 \exp\left( -2m t_j^2 \right) \right]^k,
\end{equation}
where each observable is measured $m$ times, and hence there are $km$ measurements done in total.
The shape of the hyperrectangle can be adjusted by choosing appropriate $t_j$.

Independent from the specific shape and origin of the experimenter's confidence region, they can exclude a separable quantum state as cause of the data with statistical significance only if the separable numerical range and the confidence region are disjoint.
Thus, generically, choosing observables such that the volume ratio between separable and standard numerical range is small provides a higher statistical significance for entanglement detection.
This is because the confidence region is more likely to lie outside the separable numerical range.
Importantly, this question is different from maximizing the number, i.e.~the volume, of entangled states that can be detected by infinite repetition of the measurements with infinite precision.

This reasoning motivates us to investigate the volume ratio of the (separable) numerical range.
As it turns out, it cannot be arbitrarily small and we provide bounds for any number of particles, local dimensions, and number of observables.
Moreover, we focus on product observables since they are easier to implement in experiments.
For two qubits and two locally traceless product observables, we provide explicit expressions for the volume of their standard and separable numerical range.

\section{Extreme instances}
In this section, we investigate minimal volume ratios of the separable numerical range compared to the standard numerical range that can be reached for given number of particles, local dimensions, and number of measurements.
More precisely, we define:
\begin{definition}
  For $k$ independent measurements on a quantum system consisting of $n$ particles and local dimensions $\bm{d} = (d_1,\dots,d_n)$, we denote the minimal volume ratio of the separable numerical range compared to the standard numerical range as
  \begin{align}
    \mu_{n,\bm{d},k} = \min_{A_1,\dots,A_k} \frac{\vol L_\Sep(A_1,\dots,A_k)}{\vol L(A_1,\dots,A_k)}.
  \end{align}
  If $d_1 = \dots = d_n = d$, we just write $d$ instead of $\bm{d}$.
  Further, if $n=2$, we omit the corresponding subscript.
  Finally, we denote the total dimension of the Hilbert space as $D = d_1 \cdots d_n$.
\end{definition}
As we discussed in section~\ref{sec:expconf}, measurements reaching the minimal volume ratio are in some sense optimal for entanglement detection in practical experiments.
Thus, we find lower bounds for $\mu_{n,\bm{d},k}$ as well as measurements with low $\mu_{n,\bm{d},k}$, consequently also providing upper bounds.

\subsection{General properties of the volume ratio}\label{sec:generalbounds}
In the following, we prove general properties of the volume ratio $\mu_{n,\bm{d},k}$ as well as a lower bound independent of the specific observables.
First, we consider the case in which the observables $A_1,\dots,A_k$ are not linearly independent.
Then, the (separable) numerical range is contained in a lower-dimensional manifold and the volume in $\dR^k$ vanishes.
However, we can still define the relative volume comparing the volumes of the manifolds.
More specifically, we have the following result.
\begin{proposition}
For Hermitian observables $A_1,\dots,A_k$, let $B_1,\dots,B_{k'}$ be a maximal linearly independent subset of $\{A_1 - \frac{\one}{D}\trace A_1,\dots,A_{k} - \frac{\one}{D}\trace A_{k}\}$. Then, it holds that
\begin{equation}
  \frac{\vol L_\Sep(A_1,\dots,A_k)}{\vol L(A_1,\dots,A_k)} = \frac{\vol L_\Sep(B_1,\dots,B_{k'})}{\vol L(B_1,\dots,B_{k'})},
\end{equation}
i.e., we can simply ignore observables that are linearly dependent.
\end{proposition}
\begin{proof}
Adding multiples of the identity to the observables corresponds to a translation of the (separable) numerical range, and hence, does not change the volume ratio.
More generally, affine transformations also do not change the relative volume; for details, see the proof of Proposition~\ref{prop:maxset} in Appendix~\ref{app:maxset}.
Let $A_j - \frac{\one}{D}$ depend linearly on the $B_1,\dots,B_k$, i.e., $A_j - \frac{\one}{D} = \sum_l x_l B_l$.
Then, we apply the transformation $a_j \rightarrow \tilde{a}_j = a_j - \sum_l x_l b_l$, where $a_j$ is the variable corresponding to the observable $A_j - \frac{\one}{D}$ and similarly for $b_j$ corresponding to $B_j$.
That means, we have
\begin{equation}
  \begin{aligned}
    \tilde{a}_j &= a_j - \sum_l x_l b_l \\
                &= \trace \left[ (A_j - \frac{\one}{D}) - \sum_l x_l B_l \right] \rho = 0,
  \end{aligned}
\end{equation}
for any state $\rho$.
Hence, we obtain the volume in the subspace given by $\tilde{a}_j = 0$ which simply lowers the dimension.
Applying this procedure iteratively proves the statement.
\end{proof}
Thus, it suffices to restrict the observables to be linearly independent, traceless, and bounded.

The main idea to obtain a general lower bound for the volume ratio $\mu_{n,\bm{d},k}$ is to compute the volume of the (separable) numerical range via integration using polar coordinates.
This idea is inspired by the approach used in Refs.~\cite{cowen,gawron2014}, however, while the authors of these works focus on simply finding all the boundary points, i.e., the extreme points whose convex hull forms the (separable) numerical range, we find the boundary point in a certain direction. 
\begin{lemma}\label{lem:volint}
  For Hermitian operators $A_1,\dots,A_k$ and a star-convex state set $X$ around the maximally mixed state, the $k$-dimensional volume of the restricted numerical range is given by
  \begin{equation}
    \begin{aligned}
      \vol L_X = \int_0^{2\pi}d\varphi \int_0^\pi d\vartheta_1 \cdots \int_0^\pi d\vartheta_{k-2} \, \frac{1}{k} R^k \prod_{j=1}^{k-2} \sin^j \vartheta_j,
    \end{aligned}
  \end{equation}
  where the radius $R(\varphi,\vartheta_1,\dots,\vartheta_{k-2})$ is determined by
  \begin{equation}
    \begin{aligned}
      R = \,&\maxover[\rho \in X]  && \trace \rho [\bm{\hat{r}} \cdot \bm{A}] \\
            &\subto && \trace \rho [\bm{\hat{\varphi}} \cdot \bm{A}] = 0, \\
            &       && \trace \rho [\bm{\hat{\vartheta}_j} \cdot \bm{A}] = 0 \text{ for } j=1,\dots,k-2,
    \end{aligned}
  \end{equation}
  and $\bm{A} = (A_1 - \frac{\one}{\trace\one}\trace A_1,\dots,A_k - \frac{\one}{\trace\one}\trace A_k)^T$. The vectors $\bm{\hat{r}},\bm{\hat{\varphi}},\bm{\hat{\vartheta}_1},\dots,\bm{\hat{\vartheta}_{k-2}}$ are the unit vectors of the $k$-dimensional polar coordinates.
\end{lemma}
The proof of Lemma~\ref{lem:volint} is shown in Appendix~\ref{app:volint}.
This result not only gives an interesting characterization of the volume, it can also be directly implemented using semi-definite programming \cite{sdp} to efficiently approximate the volume if the set $X$ can be characterized using semi-definite and linear constraints which, e.g., is the case if $X$ is the set of all quantum states or the set of quantum states with positive partial transpose.
In the case of two qubits, the separable states are exactly those with positive partial transpose \cite{pptsufficient}, which allows for efficient numerical treatment via semi-definite programming.

When we are only interested in the relative volume $\vol L_X / \vol L_Y$, we can obtain a lower bound by comparing the optimizations of the radii.
In particular, for $X$ being the set of separable states and $Y$ being the set of all quantum states, we obtain:
\begin{theorem}\label{thm:intbound}
  For $k$ observables and $n$-partite quantum systems with local dimensions $\bm{d} = (d_1,\dots,d_n)$ and total dimension $D = d_1 \cdots d_n$, the relative volume of the numerical range restricted to separable states compared to all quantum states is lower bounded by
  \begin{equation}
    \mu_{n,\bm{d},k} \ge \left[ \frac{b}{D} \sqrt{\frac{D-1}{D-b^2}} \right]^k,
  \end{equation}
  where
  \begin{equation}
    b = \sqrt{\frac{D^n}{(2D-1)^{n-2}(D^2-1)+1}}.
  \end{equation}
  Moreover, for a bipartite $d\times d$-system, i.e., $n=2$ and $d_1=d_2=d$, it holds that
  \begin{equation}
    \mu_{d,k} \ge \frac{1}{\left(d^2-1\right)^k}.
  \end{equation}
\end{theorem}
\begin{proof}
  For an $n$-partite quantum system $\rho$ with local dimensions $d_1,\dots,d_n$ and total dimension $D = d_1 \cdots d_n$, the state $\rho = (1-\epsilon) \frac{\one}{D} + \epsilon \sigma$ is fully separable for any state $\sigma$ if $\epsilon \le \frac{b}{D} \sqrt{\frac{D-1}{D-b^2}}$, where $b = \sqrt{\frac{D^n}{(2D-1)^{n-2}(D^2-1)+1}}$  \cite{gurvits2005, hildebrand2007}.
  Also, for bipartite systems with $d_1=d_2=d$, the same is true if $\epsilon \le 1/(d^2-1)$ \cite{zyczkowski1998volume,gurvits2002}.
  Let $\rho^*(\varphi,\vartheta_1,\dots,\vartheta_{k-2})$ be the optimal state in the maximization that determines $R(\varphi,\vartheta_1,\dots,\vartheta_{k-2})$ in Lemma~\ref{lem:volint} for $X$ being the set of all quantum states with objective value $R^*(\varphi,\vartheta_1,\dots,\vartheta_{k-2})$.
  Then, the state $\tilde{\rho} = \epsilon \rho^* + (1-\epsilon) \frac{\one}{D}$ with maximal $\epsilon$ such that full separability can be guaranteed with the above results is a feasible point of the corresponding optimization with $X$ being the set of fully separable states with objective value $\epsilon R^*$ since $\trace (A_j - \frac{\one}{\trace\one}\trace A_j) \frac{\one}{D} = 0$.
  Together with Lemma~\ref{lem:volint}, this proves the theorem.
\end{proof}

In the simplest example, i.e., a single Hermitian operator $A$ and a two-qubit system $\rho$, Theorem~\ref{thm:intbound} gives a lower bound of $1/3$ for the minimal relative volume $\mu_{2,1}$.

Another special case is the scenario in which the measurement results determine the underlying quantum state uniquely.
Although the volume ratio of the (separable) numerical range is not directly related to the share of nondetectable states, in the case of quantum state tomography, they coincide.
\begin{proposition}\label{prop:maxset}
Let $A_1,\dots,A_{D^2-1}$ be Hermitian observables such that the translated operators $A_1 - \frac{\one}{D}\trace A_1,\dots,A_{D^2-1} - \frac{\one}{D}\trace A_{D^2-1}$ are linearly independent. Then, the relative volume of the (separable) numerical range is given by
\begin{equation}
  \frac{\vol L_\Sep(A_1,\dots,A_{D^2-1})}{\vol L(A_1,\dots,A_{D^2-1})} = \frac{\vol_{HS} SEP}{\vol_{HS} ALL},
\end{equation}
where $\vol_{HS} SEP$ and $\vol_{HS} ALL$ denote the volume of separable and all states, respectively, w.r.t.~the Hilbert-Schmidt norm.
\end{proposition}
This result comes from the fact that, in this case, the space of measurement results is a mere affine transformation of the quantum state space with respect to the Hilbert-Schmidt norm; a detailed proof is shown in Appendix~\ref{app:maxset}.
In the case of two qubits and $n = 15$, the volume ratio is strongly believed to be $\frac{\vol_{HS} SEP}{\vol_{HS} ALL} = \frac{8}{33}$ \cite{sepvol1,sepvol2,sepvol3,sepvol4}.
In general, we can find a lower bound on the volume ratio considering the $\epsilon$-ball around the maximally mixed state \cite{gurvits2005,hildebrand2007} compared to the volume of all states \cite{volALL}.
From this result for $n = D^2 - 1$, it might be possible to obtain bounds for a lower number of observables as well since this corresponds to some projection of convex bodies; see Appendix~\ref{app:projconvbody}.

\subsection{Two qubits}
In the simplest scenario, i.e., a single measurement of a two-qubit quantum system, Theorem~\ref{thm:intbound} gives a bound $\mu_{2,1} \ge \frac{1}{3}$ for the minimal volume ratio.
To obtain a better bound, we use so-called absolutely separable states, whose separability can be inferred from the eigenvalues of the density matrix \cite{abssep,ishizaka2000,verstraete2001}.
\begin{proposition}\label{prop:abssepbound}
It holds that $\mu_{2,1} \ge \sqrt{2} - 1 \approx 0.41$. Moreover, this is the best bound achievable when only absolutely separable states are considered.
\end{proposition}
Proof of this proposition is given in Appendix~\ref{app:abssepbound}.
However, extensive numerical investigation suggests that also this bound is not tight and leads us to formulate the following conjecture.
\begin{conjecture}
For a single measurement on a two-qubit system, the minimal volume ratio is $\mu_{2,1} = \frac{1}{2}$.
\end{conjecture}
This value is for example reached by $A = \ket{\phi^+}\bra{\phi^+}$ with eigenvalues 0 and 1, being the projector onto the maximally entangled state $\ket{\phi^+} = \frac{1}{\sqrt{2}}(\ket{00}+\ket{11})$.
From the given Schmidt decomposition, it is obvious that the maximal overlap of a product state with the maximally entangled state is $\frac{1}{2}$; also, the overlap of the product state $\ket{01}$ with $\ket{\phi^+}$ is 0.
Thus, it follows that $\mu_{2,1} \le \frac{1}{2}$.

\begin{figure}[t!]
  \centering
  \includegraphics[width=0.7\linewidth]{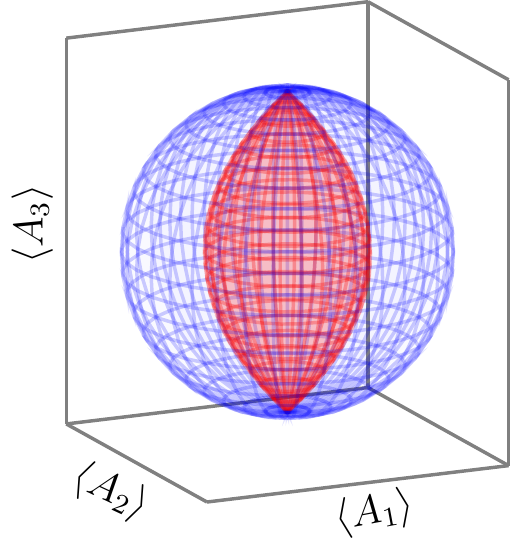}
  \caption{This figure shows the (separable) numerical range for a two-qubit quantum system and observables $A_1 = 0 \oplus X \oplus 0$, $A_2 = 0 \oplus Y \oplus 0$, and $A_3 = 0 \oplus Z \oplus 0$. The relative volume is given by $\frac{1}{5}$ and hence, $\mu_{2,3} \le \frac{1}{5}$. Also, when only the measurements $A_1$ and $A_2$ are considered, the relative volume is given by $\frac{1}{4}$ and hence, $\mu_{2,2} \le \frac{1}{4}$.}
  \label{fig:223_1over5}
\end{figure}

Concerning multiple observables, the bound in Theorem~\ref{thm:intbound} decreases exponentially with the number of measurements.
From Proposition~\ref{prop:maxset}, it is clear that, in contrast to the bounds found, the actual minimal volume ratio $\mu_{2,k}$ does not decrease exponentially with $k$.
Indeed, we find an instance of three measurements proving that $\mu_{2,3} \le \frac{1}{5} < \frac{8}{33}$ and hence, $\mu_{2,k}$ is a nonmonotonic function of $k$.
This example is given by the observables $A_1 = 0 \oplus X \oplus 0$, $A_2 = 0 \oplus Y \oplus 0$, and $A_3 = 0 \oplus Z \oplus 0$, where $X$, $Y$, and $Z$ are the Pauli matrices
\begin{align}
  X = \begin{pmatrix} 0 & 1 \\ 1 & 0 \end{pmatrix}, \quad
  Y = \begin{pmatrix} 0 & -i \\ i & 0 \end{pmatrix}, \quad
  Z = \begin{pmatrix} 1 & 0 \\ 0 & -1 \end{pmatrix},
\end{align}
and the symbol $\oplus$ denotes the direct sum of matrices where the number 0 is understood as a 1-by-1 matrix.
The resulting (separable) numerical range is visualized in Fig. \ref{fig:223_1over5}.
From the structure of the observables, it is clear that the joint numerical range is a Bloch ball since only the subspace of nonzero eigenvalues of the $A_j$ is relevant.
For separable states, the local unitaries $U_1 \otimes U_2$, where
\begin{align}
  U_1 &= \ket{0}\bra{0} + e^{i\varphi/2} \ket{1}\bra{1}, \\
  U_2 &= \ket{0}\bra{0} + e^{-i\varphi/2} \ket{1}\bra{1},
\end{align}
leave $A_3$ invariant while continuously transforming $A_1$ and $A_2$ as $A_1 \rightarrow \cos\varphi A_1 + \sin\varphi A_2$ and $A_2 \rightarrow -\sin\varphi A_1 + \cos\varphi A_2$, respectively.
Thus, the separable numerical range is symmetric w.r.t.~rotations around the axis of the third measurement, but it is not symmetric w.r.t.~other rotations.
As we are considering the three Pauli matrices, this asymmetry might seem counter-intuitive at first glance, however, while the eigenvectors with corresponding nonzero eigenvalues are product states for $A_3$, this is not true for $A_1$ and $A_2$, which explains the difference in symmetries for separable and all quantum states.
Also, restricting just to measurements $A_1$ and $A_2$ gives two concentric circles for the (separable) numerical range with a volume ratio of $\frac{1}{4}$ and hence, $\mu_{2,2} \le \frac{1}{4}$.
Detailed volume calculations are shown in Appendix~\ref{app:2quintances}.

\subsection{Product observables}
\begin{figure}[t!]
  \centering
  \includegraphics[width=0.9\linewidth]{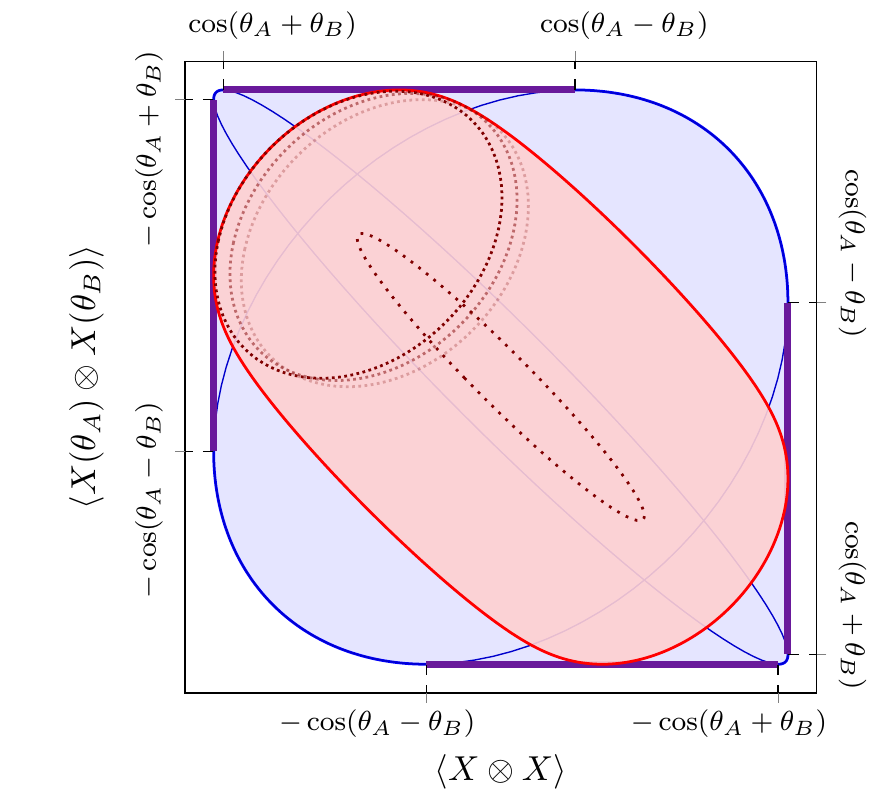}
  \caption{Illustration of the proof of Theorem~\ref{thm:222sep}, where $X(\theta) = X\cos\theta + Z\sin\theta$.
  For two locally traceless two-qubit observables, the separable numerical range is the Minkowski sum of two ellipses (red), while the standard numerical range is the convex hull of two other ellipses (blue, with the flat parts highlighted in purple).
  In this plot, $\theta_A = 3\pi/4$ and $\theta_B = \pi/3$.}
  \label{fig:prod_obs}
\end{figure}
In practice, it is much harder to implement highly entangled measurements compared to local, i.e. product, observables.
That is why, in the following, we focus on such simpler observables.
Product observables, i.e. $A=B_1 \otimes B_2$, are easy to measure in spatially separated laboratories or on a composite quantum system consisting of separate particles, and therefore provide a physically well motivated subset of all possible observables.
For a single product observable, separable and standard numerical range are obviously identical as the eigenvectors are product states.
Thus, let us consider the simplest non-trivial case which is a two-qubit system and two product observables $A_1,A_2$.

Moreover, we restrict ourselves to locally traceless observables, i.e., after applying suitable local unitaries, which apparently do not affect the relative volume, we have
\begin{align}
  A_1 &= X \otimes X, \label{eq:ltobs1} \\
  A_2 &= \left( \cos\theta_A X + \sin\theta_A Z \right) \otimes \left( \cos\theta_B X + \sin\theta_B Z \right). \label{eq:ltobs2}
\end{align}
For local measurements, the restriction to locally traceless observables corresponds to a constant offset of the measurement results which is trivial from an experimenter's point of view.

We denote the minimal volume ratio for (locally traceless) product measurements as $\mu_{n,\bm{d},k}^{\otimes}$ ($\mu_{n,\bm{d},k}^{\otimes,\text{LT}}$) where $n$ is the number of particles, $\bm{d}$ the vector of local dimensions, and $k$ the number of observables.
\renewcommand{\arraystretch}{1.4}
\begin{table}[t!]
  \centering
  \begin{tabular}{|c||c|c||c|c|}
  \hline
  $k$ & \multicolumn{2}{c||}{$\mu_{2,k}$} & \multicolumn{2}{c|}{$\mu_{2,k}^{\otimes,\text{LT}}$} \\[3pt]
  & low. bound & upp. bound & low. bound & upp. bound \\ \hhline{|=#=|=#==|}
  $~1~$ & $\sqrt{2} - 1$ & $~1/2~$ & \multicolumn{2}{c|}{$1$} \\ \hline
  $~2~$ & $~1/9~$ & $~1/4~$ & \multicolumn{2}{c|}{$1/2^*$} \\ \hline
  $~3~$ & $~1/27~$ & $~1/5~$ & $~1/27~$ & $~1/3~$ \\ \hline
  $~4~$ & $~1/81~$ & $~1/6~$ & $~1/81~$ & $~1/6~$ \\ \hline
  $\cdots$ & $\cdots$ & $\cdots$ & $\cdots$ & $\cdots$ \\ \hline
  $~15~$ & \multicolumn{2}{c||}{$8/33$} & \multicolumn{2}{c|}{$8/33$} \\ \hline
  \end{tabular}
  \caption{This table shows lower and upper bounds for $\mu_{2,k}$ and $\mu_{2,k}^{\otimes,\text{LT}}$ for different numbers of measurements $k$ on a two-qubit system. The starred value is obtained partially via numerical two-parameter optimization.}
  \label{tbl:2qusum}
\end{table}
\renewcommand{\arraystretch}{1}
Surprisingly, we can compute the volume ratio for given locally traceless observables explicitly. 
\begin{theorem}\label{thm:222sep}
For locally traceless two-qubit product observables $A_1$ and $A_2$ written in the standard form as in Eqs.~(\ref{eq:ltobs1},\,\ref{eq:ltobs2}), the volume ratio is 
\begin{equation}
  \begin{aligned}
    \frac{\vol L_\Sep}{\vol L} = &\frac{\pi}{8} \left[ \left( \left|\sin\theta_-\right|+\left|\sin\theta_+\right| \right) - \frac{ \tilde{F}(\theta_-,\theta_+)}{\tilde{T}(\theta_-,\theta_+)} \right] /  \\
                &\left[ \cos\theta_- - \cos\theta_+ + G_-(\theta_-) + G_+(\theta_+) \right],
  \end{aligned}
\end{equation}
where $\theta_\pm = \theta_A \pm \theta_B$ as well as
\begin{align}
  G_-(\theta_-) &= |\frac{\theta_-}{2} \sin\theta_-|, \\
  G_+(\theta_+) &= |\left( \frac{\theta_+}{2}-\frac{\pi}{2} \right) \sin\theta_+|, \\
  \tilde{F}(\theta_-,\theta_+) &= F(\theta_-,\theta_+) - F(\theta_+,\theta_-), \\
  \tilde{T}(\theta_-,\theta_+) &= T(\theta_-,\theta_+) - T(\theta_+,\theta_-),
\end{align}
and the functions $F$ and $T$ are given by
\begin{align}
  F(x,y) &= \left| \sin\frac{x}{2}\cos\frac{y}{2} \right| \left[ K(1-T^2) - E(1-T^2) \right], \\
  T(x,y) &= \left| \frac{\tan\frac{x}{2}}{\tan\frac{y}{2}} \right|,
\end{align}
where $K(\cdot)$ and $E(\cdot)$ are the elliptic integrals of the first and second kind, respectively.
\end{theorem}
Indeed, we prove that the separable numerical range is given by the Minkowski sum of two ellipses, while the standard numerical range is the convex hull of two other ellipses; see Fig.~\ref{fig:prod_obs}.
The proof is shown in Appendix~\ref{app:222sep}, where we compute both shapes and volumes explicitly.
Using this result, numerical two-parameter optimization shows that $\mu_{2,2}^{\otimes,\text{LT}} = \frac{1}{2}$.

We also compute the volume ratio for certain instances of more than two measurements.
More precisely, the observables $X\otimes X$, $X\otimes Y$, and $Z\otimes Z$ yield a ratio of $1/3$, adding also the observable $Y\otimes Z$ yields a ratio of $1/6$; see Appendix~\ref{app:2quintances} for details.

In Table~\ref{tbl:2qusum}, we summarize our results bounding the minimal volume ratios $\mu_{2,k}$ and $\mu_{2,k}^{\otimes,\text{LT}}$. The upper bounds for $\mu_{2,1}$ and $\mu_{2,2}$ are presumably tight.

\section{Average cases}

To compare our findings with the \emph{generic} case, one may ask the following question: For Hermitian operators $X_1,\dots,X_{k}$ acting on a bipartite $d\times d$ system, drawn i.i.d. from Gaussian Orthogonal Ensemble (GOE), what is the average volume ratio,
\begin{equation}
 \tau_{d,k} = \left\langle\frac{\vol L_\Sep(X_1,\dots,X_{k})}{\vol L(X_1,\dots,X_{k})} \right\rangle_{X_1,\dots,X_{k} \sim \text{GOE}}.
 \label{eq:meanvolratio}
\end{equation}
\renewcommand{\arraystretch}{1.4}
\begin{table}[t!]
	\centering
	\begin{tabular}{|c|c|c||c|}
		\hline
		~$d^2$~ & ~$\langle |\lambda^\otimes_\text{min}| \rangle_{\text{GOE}}$~ & ~$\langle|\lambda_{\text{min}}|\rangle_{\text{GOE}}$~ & ~$\langle {\lambda^\otimes_{\text{min}}}/{\lambda_{\text{min}}} \rangle_{\text{GOE}}$~ \\[3pt] \hhline{|=|=|=#=|}
		$2^2$ & $1.213(4)$ & $1.364(4)$ & $0.8871(4)$ \\
		\hline
		$3^2$ & $1.3185(3)$ & $1.6542(2)$ & $0.7981(8)$ \\
		\hline
		$4^2$ & $1.2862(1)$ & $1.7738(2)$ & $0.7265(6)$ \\
		\hline
	\end{tabular}
	\caption{Mean absolute values $\langle |\lambda^\otimes_\text{min}| \rangle$ and $\langle|\lambda_{\text{min}}|\rangle$ for GOE matrices of size $d^2$ based on 15000 samples each. The mean ratio $\langle {\lambda^\otimes_{\text{min}}}/{\lambda_{\text{min}}} \rangle$ decreases with dimension --- a sign that the set of separable states becomes increasingly smaller in comparison with the entire set of quantum states.}
	\label{tab:lambdas}
\end{table}
\renewcommand{\arraystretch}{1}

In order to answer this question and to set the bounds derived in this article in a context, we have numerically probed the separable and joint numerical ranges for different numbers of operators $k$ and subsystem dimensions $d$. 
One of the variables related to the volume ratio is the ratio of the \emph{minimal separable expectation value},
\begin{equation}
\lambda^\otimes_{\min}(X) = \min_{\ket{\alpha}, \ket\beta} \braket{\alpha\otimes\beta | X | \alpha\otimes\beta},
\end{equation}
to the minimal eigenvalue $\lambda_{\text{min}}(X)$.
Note that for matrices randomly distributed w.r.t.~the Gaussian Orthogonal Ensemble, the average of this ratio is equal to the average ratio of maximal expectation values due to symmetry.
Both describe the ratio between separable and standard numerical radius of a single bipartite observable $X$.
Although the ratio ${\lambda^\otimes_{\min}}/{\lambda_{\min}}$ is not \emph{exactly} equal to the volume ratio $\vol L_{\Sep}/\vol L$, the joint and separable numerical ranges of matrices drawn from GOE tend to be approximately symmetric implying an approximate relation.
Even for a larger number of measurements, we expect the volume ratio and the ratio of minimal (separable) eigenvalues to be correlated.
In particular, if both numerical ranges were spheres centered at the origin, which indeed is the limiting case for large system sizes $d$, the volume ratio would satisfy
\begin{equation}
\left\langle \frac{\vol L_{\Sep}}{\vol L} \right\rangle_{\text{GOE}} \approx \left\langle \frac{\lambda^\otimes_{\min}}{\lambda_{\min}} \right\rangle _{\text{GOE}}^k.
\label{eqn:goeassumption}
\end{equation}

We numerically compute $\lambda_{\min}$ and $\lambda^\otimes_{\min}$ for two qubits, two qutrits and two ququarts based on $15 000$ samples and for higher dimensions using smaller samples. The numerical estimation of $\lambda^\otimes_{\min}$ is based on the observation that
\begin{equation}
	\label{eqn:lotimesexpr}
	\lambda^\otimes_{\text{min}}(X) = \min_{\ket{\alpha}\in H_d} \lambda_{\text{min}}\left(\text{Tr}_A\left[\ketbra{\alpha}{\alpha}\otimes\mathbb{1}\right]X\right), 
\end{equation}
allowing for minimization of the minimal eigenvalue of a lower dimensional operator with respect to the Hilbert space of a single operator \cite{czartowski2019}. The results are listed in Table~\ref{tab:lambdas} and Appendix~\ref{app:minsepexpval}.

\begin{figure}[t!]
  \centering
  \includegraphics[width=0.9\linewidth]{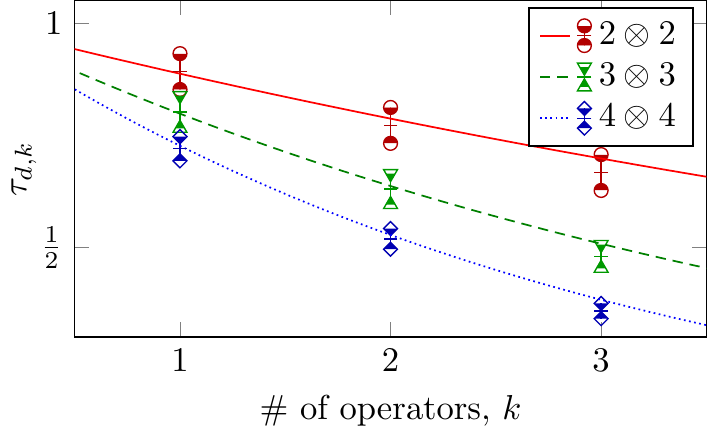}
  \caption{Comparison of the volume ratios $\tau_{d,k}$ for different bipartite $d\times d$ systems and $k$ operators drawn from the Gaussian Orthogonal Ensemble (symbols with error bars) with  the limiting case expression defined in Eq.~\eqref{eqn:goeassumption} represented by the curves. While the expression assumptions are not met --- the numerical ranges are far from being spherical for low local dimensions $d$ --- it is unexpectedly accurate for low $k$.}
  \label{fig:jnr-goe-scaling}
\end{figure}

Computing the volume ratio involves calculating the separable and standard joint numerical range for several input operators.
The volumes of the separable numerical ranges were approximated numerically using an iterative algorithm probing product states which are likely to contribute to the boundary shape. The calculations were run until a sufficient number of points on the boundary had been generated. Due to the challenging complexity, only $\sim 10^2$ volume ratios for every dimension were calculated, the results are shown in Table~\ref{tab:ratios} and compared to the approximation from Eq.~(\ref{eqn:goeassumption}) in Fig.~\ref{fig:jnr-goe-scaling}. 

\renewcommand{\arraystretch}{1.4}
\begin{table}[t!]
	\centering
	\begin{tabular}{|c|c|c|c|c|}
		\hline
		\backslashbox{$k$}{$d^2$} & $2^2$ & $3^2$   & $4^2$ \\
		\hline
		~1~ & ~$0.892(8)$~ & ~$0.802(6)$~ & ~$0.720(5)$~ \\
		\hline
		~2~ & ~$0.772(9)$~ & ~$0.630(6)$~ & ~$0.519(5)$~ \\
		\hline
		~3~ & ~$0.667(9)$~ & ~$0.480(4)$~ & ~$0.3581(3)$~ \\
		\hline
	\end{tabular}
	\caption{Volume ratios $\tau_{d,k}$ as defined in Eq. \eqref{eq:meanvolratio} of separable and standard numerical ranges as a function of the number of operators $k$ drawn from Gaussian Orthogonal Ensemble for different bipartite systems: two qubits, two qutrits, and two ququarts.}
	\label{tab:ratios}
\end{table}
\renewcommand{\arraystretch}{1}

They suggest that generically (in the terms of GOE) the volume ratios do not reach the level of the examples we found to provide upper bounds to $\mu_{d,k}$. Numerically, the asymptotic expression for large system sizes provided by Eq. \eqref{eqn:goeassumption} seems to serve as an upper bound for the average volume ratio for a low number of operators $k$.

\section{Conclusion}
Considering the (separable) numerical range is the most general concept for entanglement detection for given observables.
The volume ratio of the separable numerical range compared to the standard numerical range indicates how difficult it is to verify the presence of entanglement with statistical significance.
More precisely, in an experiment, the measurement data imply a confidence region for the expectation values of the underlying state when the given observables are measured.
To detect entanglement, this confidence region and the separable numerical range have to be disjoint, which is generically more likely for smaller volume ratios.
We provide a general lower bound for any number of particles, local dimensions, and number of measurements and consider extreme cases.
In the case of two qubits and a single measurement, we examine the numerical range generated by absolutely separable states and obtain a lower bound for the volume ratio of $\sqrt{2} - 1$.
Numerical investigations lead us to conjecture that the minimal volume ratio is indeed $1/2$, thus it would be desirable to close this gap for the most basic case in the future.
Furthermore, we focus on product observables which are easier to measure in experiments.
For two locally traceless two-qubit observables, we explicitly provide the volumes for the (separable) numerical range.
Lastly, numerical investigations into generic volume ratios help to put the obtained results into a wider context.

Besides experimental entanglement detection, our results also give insight into the geometry of the set of (separable) quantum states and their relation.
The (separable) numerical range is an affine transformation of a projection of the set of separable or general quantum states, i.e., a lower-dimensional shadow.
This leads to interesting mathematical questions about volume ratios of shadows of convex bodies.
In the future, we hope to extend our results and consider more quantum phenomena such as genuine multipartite entanglement, steering, and nonlocality, and their respective volume ratios.

\begin{acknowledgments}
  We would like to thank Otfried G{\"u}hne, Zbigniew Pucha{\l}a, Nikolai Wyderka, and Xiao-Dong Yu for fruitful discussions.
  This work was supported by the Deutsche Forschungsgemeinschaft (DFG, German Research Foundation, project numbers 447948357 and 440958198),
  the Sino-German Center for Research Promotion (Project M-0294),
  the ERC (Consolidator Grant 683107/TempoQ), 
  and the House of Young Talents Siegen.
  Also, financial support by the Polish National Science Center under the grant numbers
  DEC-2015/18/A/ST2/00274 and 2019/35/O/ST2/01049 and
  by the Foundation for Polish Science under the Team-Net NTQC project is gratefully acknowledged.
\end{acknowledgments}

\vspace{2em}
\onecolumngrid
\appendix

\section{Proof of Lemma~\ref{lem:volint}}\label{app:volint}
\setcounter{definition}{3}
\begin{lemma}
  For Hermitian operators $A_1,\dots,A_k$ and a star-convex state set $X$ around the maximally mixed state, the $k$-dimensional volume of the restricted numerical range is given by
  \begin{equation}
    \begin{aligned}
      \vol L_X = \int_0^{2\pi}d\varphi \int_0^\pi d\vartheta_1 \cdots \int_0^\pi d\vartheta_{k-2} \, \frac{1}{k} R^k \prod_{j=1}^{k-2} \sin^j \vartheta_j,
    \end{aligned}
  \end{equation}
  where the radius $R(\varphi,\vartheta_1,\dots,\vartheta_{k-2})$ is determined by
  \begin{equation}
    \begin{aligned}
      R = \,&\maxover[\rho \in X]  && \trace \rho [\bm{\hat{r}} \cdot \bm{A}] \\
            &\subto && \trace \rho [\bm{\hat{\varphi}} \cdot \bm{A}] = 0, \\
            &       && \trace \rho [\bm{\hat{\vartheta}_j} \cdot \bm{A}] = 0 \text{ for } j=1,\dots,k-2,
    \end{aligned}
  \end{equation}
  and $\bm{A} = (A_1 - \frac{\one}{\trace\one}\trace A_1,\dots,A_k - \frac{\one}{\trace\one}\trace A_k)^T$. The vectors $\bm{\hat{r}},\bm{\hat{\varphi}},\bm{\hat{\vartheta}_1},\dots,\bm{\hat{\vartheta}_{k-2}}$ are the unit vectors of the $k$-dimensional polar coordinates.
\end{lemma}
\begin{proof}
  The volume $V$ of a $k$-dimensional region that is star-convex around the origin is given via integration as
  \begin{equation}  
    \begin{aligned}
      V &= \int_0^{2\pi}d\varphi \int_0^\pi d\vartheta_1 \cdots \int_0^\pi d\vartheta_{k-2} \int_0^{R} dr \, r^{k-1} \,\prod_{j=1}^{k-2} \sin^j \vartheta_j \\
        &= \int_0^{2\pi}d\varphi \int_0^\pi d\vartheta_1 \cdots \int_0^\pi d\vartheta_{k-2} \, \frac{1}{k} R^k \,\prod_{j=1}^{k-2} \sin^j \vartheta_j,
    \end{aligned}
  \end{equation}
  using $k$-dimensional polar coordinates.
  Obviously, we have $\vol L(A_1,\dots,A_k) = \vol L(A_1 - \frac{\one}{\trace\one}\trace A_1,\dots,A_k - \frac{\one}{\trace\one}\trace A_k)$ since the transformation merely translates the restricted joint numerical range.
  The restricted numerical range $L_X(A_1 - \frac{\one}{\trace\one}\trace A_1,\dots,A_k - \frac{\one}{\trace\one}\trace A_k)$ is a star-convex set around the origin as $X$ is star-convex around the maximally mixed state.
  Hence, we can calculate its volume using the above geometric formula.
  The radius $R(\varphi,\vartheta_1,\dots,\vartheta_{k-2})$ is apparently given by
  \begin{equation}
    \begin{aligned}
      R(\varphi,\vartheta_1,\dots,\vartheta_{k-2}) = \,&\maxover[\rho \in X]  && \trace \rho [\bm{\hat{r}}(\varphi,\vartheta_1,\dots,\vartheta_{k-2}) \cdot \bm{A}] \\
            &\subto && \trace \rho [\bm{\hat{\varphi}}(\varphi,\vartheta_1,\dots,\vartheta_{k-2}) \cdot \bm{A}] = 0, \\
            &       && \trace \rho [\bm{\hat{\vartheta}_j}(\varphi,\vartheta_1,\dots,\vartheta_{k-2}) \cdot \bm{A}] = 0 \text{ for } j=1,\dots,k-2,
    \end{aligned}
  \end{equation}
  where $\bm{A} = (A_1 - \frac{\one}{\trace\one}\trace A_1,\dots,A_k - \frac{\one}{\trace\one}\trace A_k)^T$ and $\bm{\hat{r}},\bm{\hat{\varphi}},\bm{\hat{\vartheta}_1},\dots,\bm{\hat{\vartheta}_{k-2}}$ are the unit vectors of the $k$-dimensional polar coordinates. The constraints of the optimization make sure that it yields the distance of the boundary to the origin in a certain direction given by the angles $\varphi,\vartheta_1,\dots,\vartheta_{k-2}$.
\end{proof}

\section{Proof of Proposition~\ref{prop:maxset}}\label{app:maxset}
\setcounter{definition}{5}

\begin{proposition}
Let $A_1,\dots,A_{D^2-1}$ be Hermitian observables such that the translated operators $A_1 - \frac{\one}{D}\trace A_1,\dots,A_{D^2-1} - \frac{\one}{D}\trace A_{D^2-1}$ are linearly independent. Then, the relative volume of the (separable) numerical range is given by
\begin{equation}
  \frac{\vol L_\Sep(A_1,\dots,A_{D^2-1})}{\vol L(A_1,\dots,A_{D^2-1})} = \frac{\vol_{HS} SEP}{\vol_{HS} ALL},
\end{equation}
where $\vol_{HS} SEP$ and $\vol_{HS} ALL$ denote the volume of separable and all states, respectively, w.r.t.~the Hilbert-Schmidt norm.
\end{proposition}
\begin{proof}
We can express any $D$-dimensional quantum state $\rho$ in terms of the generalized Gell-Mann matrices \cite{gellmann}, i.e.,
\begin{equation}
\rho_{\bm{\xi}} = \frac{\one}{D} + \sum_{j=1}^{D^2-1} \xi_j G_j, 
\end{equation}
where the $\xi_j \in \dR$ are real coefficients and the $G_j$ form a Hermitian ($G_j^\dagger = G_j$), orthonormal ($\trace G_i G_j = \delta_{ij}$) and traceless ($\trace G_j = 0$) basis. The distance between two quantum states can be measured using the Hilbert-Schmidt norm $\norm{A}_{HS} = \sqrt{\trace A^\dagger A}$. We obtain
\begin{equation}
  \begin{aligned}
    \norm{\rho_{\bm{\xi}} - \rho_{\bm{\eta}}}_{HS} &= \sqrt{\trace\left[(\rho_{\bm{\xi}} - \rho_{\bm{\eta}})^2\right]} \\
                                                   &= \sqrt{\sum_{i,j=1}^{D^2-1} (\xi_i - \eta_i)(\xi_j - \eta_j) \trace G_i G_j} \\
                                                   &= \sqrt{\sum_{j=1}^{D^2-1} (\xi_j - \eta_j)^2} \\
                                                   &= \norm{\bm{\xi} - \bm{\eta}}
  \end{aligned}
\end{equation}
for the distance between the quantum states $\rho_{\bm{\xi}} = \frac{\one}{D} + \sum_{j=1}^{D^2-1} \xi_j G_j$ and $\rho_{\bm{\eta}} = \frac{\one}{D} + \sum_{j=1}^{D^2-1} \eta_j G_j$. This is the same as the Euclidean distance when we consider the $\xi_j$ and $\eta_j$ as coordinates in $\dR^{D^2-1}$.

Because the $\tilde{A}_j = A_j - \frac{\one}{D}\trace A_j$ are linearly independent and traceless, there exists a dual basis \cite{lebedev2010} $\{B_j\}_{j=1,\dots,D^2-1}$, i.e., $\trace \tilde{A}_i B_j = \delta_{ij}$, and an invertible matrix $\Lambda$ such that $G_i = \sum_j \Lambda_{ij} B_j$. Furthermore, we have that
\begin{equation}
  \begin{aligned}
    \rho_{\bm{\xi}} &= \frac{\one}{D} + \bm{\xi} \cdot \bm{G} \\
                    &= \frac{\one}{D} + \bm{\xi} \cdot \Lambda \bm{B} \\
                    &= \frac{\one}{D} + (\Lambda^T \bm{\xi}) \cdot \bm{B},
  \end{aligned}
\end{equation}
where $\Lambda^T$ is the tranpose of $\Lambda$, which is also invertible.
The coefficients $(\Lambda^T \bm{\xi})_j$ are the coordinates of $\rho$ in the space $L(\tilde{A}_1,\dots,\tilde{A}_{D^2-1})$ because we used the dual basis.

In general, a coordinate transformation $(v_1,\dots,v_k) = \varphi(u_1,\dots,u_k)$ leads to a change of volume integrals
\begin{equation}
  \int_{\bm{v} \in \varphi(U)} f(\bm{v}) d^k v = \int_{\bm{u} \in U} f(\varphi(\bm{u})) |\det(D\varphi)(\bm{u})| d^k u,
\end{equation}
where $\det(D\varphi)$ is the determinant of the Jacobi matrix of $\varphi$.
In this case, $\varphi(\bm{u}) = (\Lambda^T)^{-1} \bm{v}$, which means that
\begin{equation}
  \begin{aligned}
  \vol_{HS} X &= \int_{\rho_{\bm{\xi}} \in X} d^k \xi \\
              &= |\det (\Lambda^T)^{-1}| \int_{\rho_{\bm{\chi}} \in \varphi^{-1}(X)} d^k \chi \\
              &= |\det (\Lambda^T)^{-1}| \vol L_X(A_1,\dots,A_{D^2-1}),
  \end{aligned}
\end{equation}
where $\bm{\chi} = \Lambda^T \bm{\xi}$ and $\varphi^{-1}(X) = \{ \rho_{\Lambda^T \bm{\xi}} ~|~ \rho_{\bm{\xi}} \in X \}$.
Note that $|\det (\Lambda^T)^{-1}| = \const. > 0$ because the transformation is invertible.
Also, since the transformation is linear, it is independent of the variables.
Hence, the relative volume does not change as we apply the same transformation independently of $X$.
\end{proof}

\section{Projection of convex bodies}\label{app:projconvbody}
\begin{figure}[t!]
  \centering
  \begin{subfigure}[t]{0.45\linewidth}
  \centering
    \includegraphics[width=\linewidth]{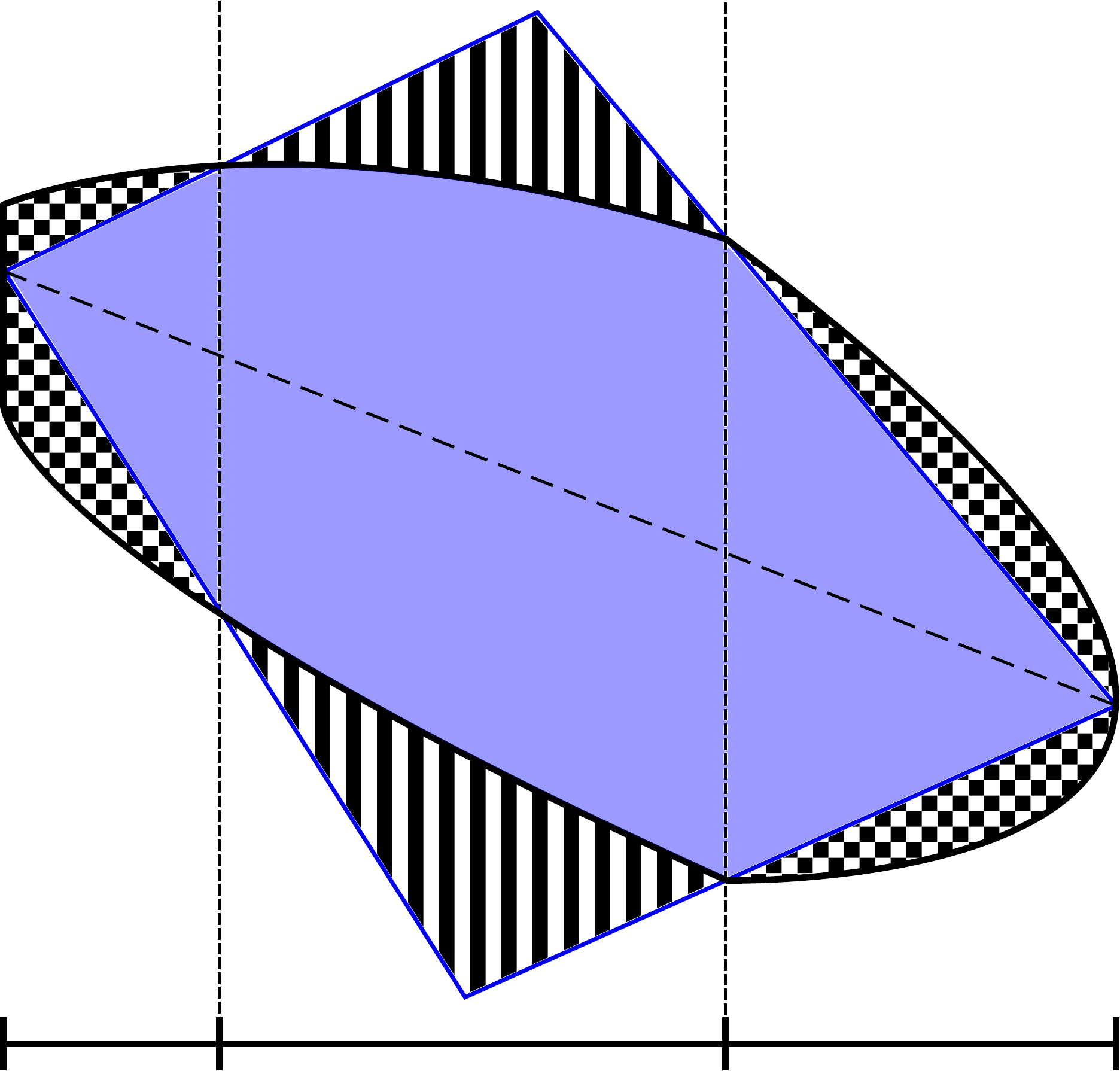}
    \subcaption{To maximize the volume ratio, the best choice for the inner convex body is an infinite rectangle bounded by the outer body. Furthermore, it is shown how the outer convex body can be transformed to a convex quadrangle with larger volume ratio while not changing the volume ratio of the projection by removing checkerboard areas and adding striped areas.}
    \label{fig:proj_bound_d2_a}
  \end{subfigure}
  \hspace*{0.05\linewidth}
  \begin{subfigure}[t]{0.45\linewidth}
  \centering
    \includegraphics[width=\linewidth]{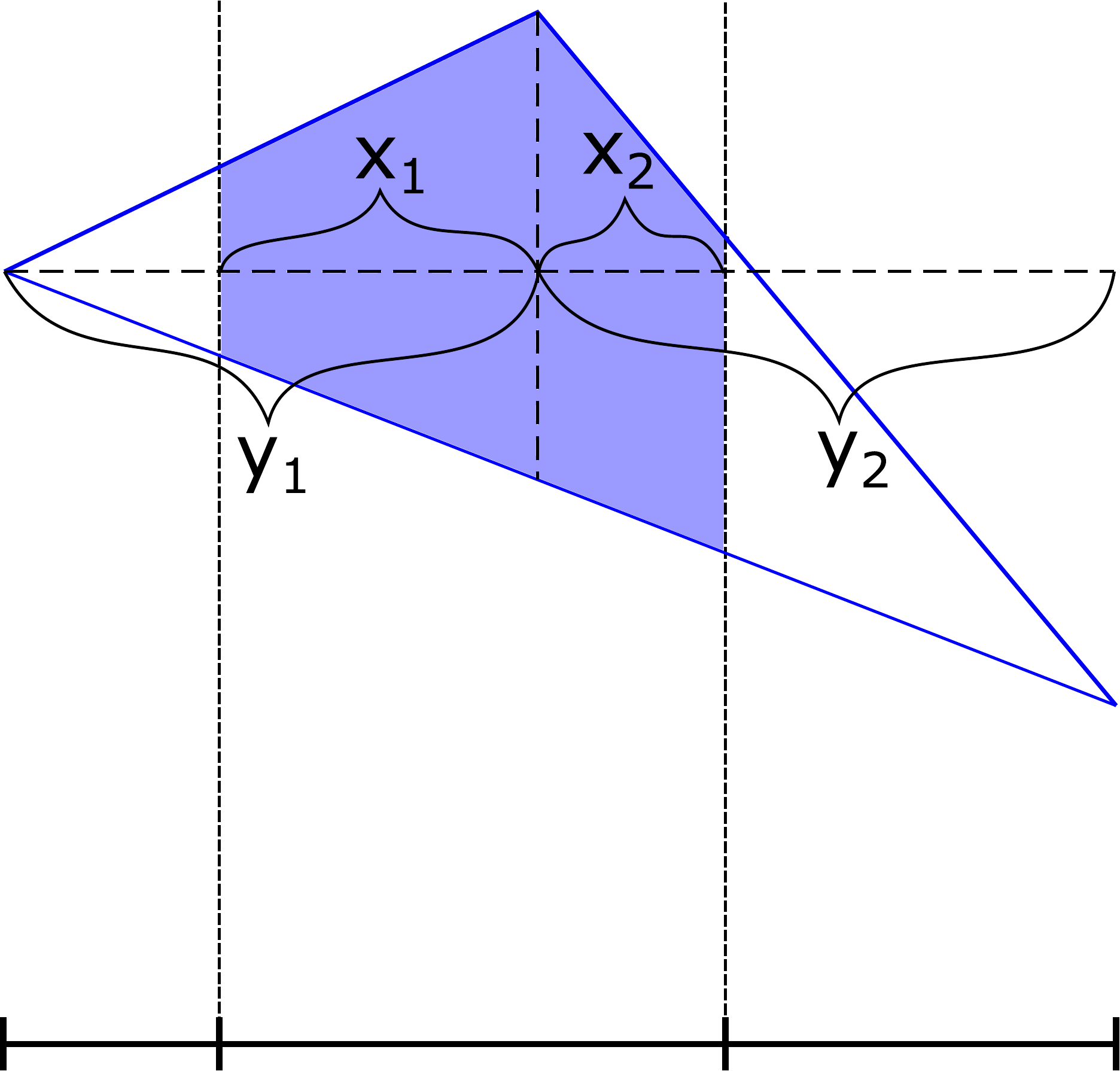}
    \subcaption{The volume ratio of the resulting shape can be optimized for both triangles separately. It is given by $\frac{\vol P_{\bm{u}}K}{\vol P_{\bm{u}}L} = 1 - \frac{(y_1 - x_1)^2}{y_1 y} - \frac{(y_2-x_2)^2}{y_2 y}$, independent from the vertical positions of the corner points, and direct optimization gives $\frac{\vol P_{\bm{u}}K}{\vol P_{\bm{u}}L} = 1-(1-\frac{x}{y})^2$ with $x=x_1+x_2$ and $y=y_1+y_2$ as the largest possible volume ratio for fixed volume ratio $\frac{x}{y}$ of the projections.}
    \label{fig:proj_bound_d2_b}
  \end{subfigure}
  \caption{Illustration of the proof of Proposition~\ref{prop:volratio_n2}.}
  \label{fig:proj_bound_d2}
\end{figure}
Starting from $D^2-1$ observables and removing some of them corresponds to taking projections onto lower dimensional subspaces.
Then, the idea is that the volume ratio of the projections cannot be arbitrarily small compared to the volume ratio of the original convex bodies.
The (very much mathematical) geometric question is: 

Given two convex bodies, i.e.~convex, compact sets, $K,L \subset \dR^k$ such that $K \subset L$ with fixed volume ratio $\frac{\vol K}{\vol L}$. What is the minimal volume ratio of projections $\min_{\bm{u_1},\dots,\bm{u_l}}\frac{\vol P_{\bm{u_1},\dots,\bm{u_l}}K}{\vol P_{\bm{u_1},\dots,\bm{u_l}}L}$, where $P_{\bm{u_1},\dots,\bm{u_l}}$ denotes the projection onto the subspace orthogonal to the vectors $\bm{u_1},\dots,\bm{u_l}$?

The question might be easier when we only consider $l=1$.
The following observation finds the optimal bound in the simplest case, i.e.~$k=2$.
\setcounter{definition}{9}
\begin{proposition}\label{prop:volratio_n2}
For two convex sets $K \subset L \subset \dR^2$ with fixed volume ratio $\frac{\vol K}{\vol L}$, it holds that
\begin{equation}
  \min_{\bm{u}\in\dR^2} \frac{\vol P_{\bm{u}}K}{\vol P_{\bm{u}}L} \ge 1 - \sqrt{1-\frac{\vol K}{\vol L}}.
\end{equation}
This bound is tight.
\end{proposition}
\begin{proof}
The proof is visualized in Fig.~\ref{fig:proj_bound_d2}.
Instead of minimizing the volume ratio of projections for given convex bodies, we consider two projections $PK \subset PL \subset \dR^1$ with given volume ratio and maximize the volume ratio of possible convex bodies $K \subset L$.
That means, we maximize the volume of $K$ while minimizing that of $L$ under the constraints.
Since $PK$ is some line segment, the maximal volume is reached by the infinite rectangle with this projection bounded by $L$ as we have $K \subset L$.
An arbitrary body $L$ can be made smaller without changing $K$ by considering two points that project onto the two end points of $PL$ and their straight connections to the points that lie on the boundary of both $K$ and $L$.
Removing all points from $L$ that are not within the set constrained by these straight lines, makes the volume of $L$ smaller while not changing the volume of $K$, thus increasing the volume ratio.
These points must have been inside $L$ before because of convexity.
Points that lie within the infinite rectangle with projection $PK$ and the area constrained by the straight lines can be added to both $L$ and $K$, increasing both volumes by the same constant, and hence, increasing the volume ratio; see Fig.~\ref{fig:proj_bound_d2_a}.
Thus, we are left with optimizing the position of a quadrangle.
Indeed, it is sufficient to separately optimize the triangles above and below the connecting line of the points that are projected onto the end points of $PL$.
Using the intercept theorem, it is easy to show that the relative volume is independent of the vertical position (compare Fig.~\ref{fig:proj_bound_d2_b}) of the corner points.
Straightforward optimization proofs that the optimal form is a triangle and $PK$ which are symmetric w.r.t.~the center of $PL$.
Then, the volume ratio is given by $\frac{\vol K}{\vol L} = 1 - (1 - \frac{\vol PK}{\vol PL})^2$.
\end{proof}
This proof motivates us to formulate the following conjecture.

\begin{figure}[t!]
  \centering
  \includegraphics[width=0.3\linewidth]{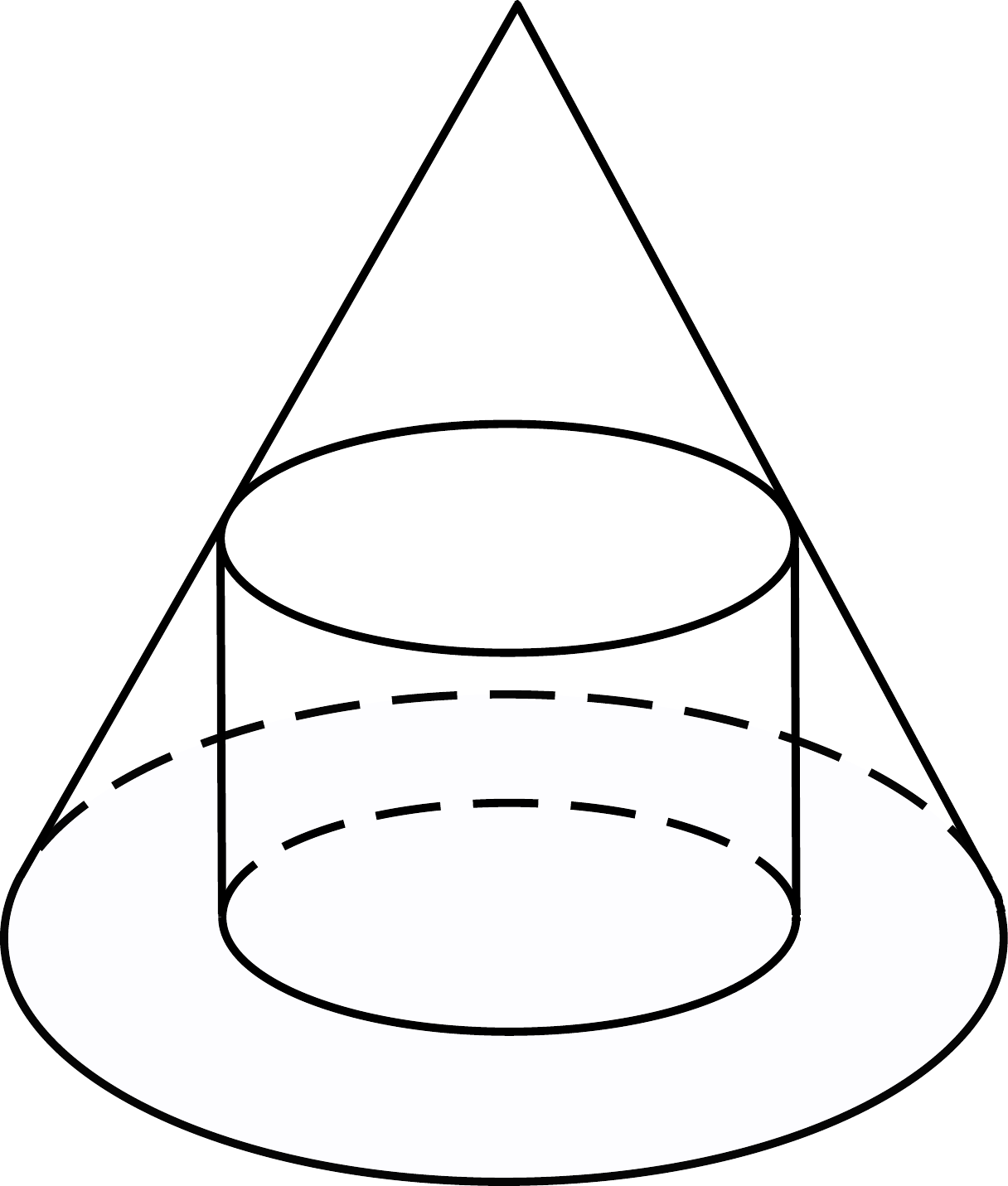}
  \caption{The outer convex body is the large cone while the inner body consists of the cylinder in addition to the small cone on top of the cylinder.
  The base of each object might be a high-dimensional ball.}
  \label{fig:proj_bound}
\end{figure}
\setcounter{definition}{10}
\begin{conjecture}\label{conj:proj_bound}
For two convex sets $K \subset L \subset \dR^k$ with fixed volume ratio $\frac{\vol K}{\vol L}$, it holds that
\begin{equation}
  \min_{\bm{u}\in\dR^k} \frac{\vol P_{\bm{u}}K}{\vol P_{\bm{u}}L} \ge c\left(\frac{\vol K}{\vol L}\right),
\end{equation}
where $0 \le c\left(\frac{\vol K}{\vol L}\right) \le 1$ is the solution to the equation
\begin{equation}
  \frac{\vol K}{\vol L} = c \left[ 1 + (k-1)\left( 1-\sqrt[k-1]{c}\right) \right].
\end{equation}
This bound is tight and can be reached by two concentric $(k-1)$-dimensional balls $PK$ and $PL$ that serve as the base for a cylindrical object and a symmetrically positioned cone, respectively; compare Fig.~\ref{fig:proj_bound}.
\end{conjecture}

Furthermore, we obtain a bound independent of $k$ by taking the limit $k \rightarrow \infty$, since
\begin{equation}
  \begin{aligned}
    \frac{\partial}{\partial n} c \left[ 1 + (k-1)\left( 1-\sqrt[k-1]{c}\right) \right] &\ge 0, \\
    \frac{\partial}{\partial c} c \left[ 1 + (k-1)\left( 1-\sqrt[k-1]{c}\right) \right] &\ge 0.
  \end{aligned}
\end{equation}
Then, we have that
\begin{equation}
  \frac{\vol K}{\vol L} = c \left( 1 - \log c \right)
\end{equation}
This configuration, however, is most likely suboptimal when $l > 1$.
Thus, better bounds could be obtained by considering the general geometrical question.

\section{Proof of Proposition~\ref{prop:abssepbound}}\label{app:abssepbound}
\setcounter{definition}{6}
\begin{proposition}
It holds that $\mu_{2,1} \ge \sqrt{2} - 1 \approx 0.41$. Moreover, this is the best bound achievable when only absolutely separable states are considered.
\end{proposition}
\begin{proof}
  Let us translate and scale the observable $A$ such that its smallest and largest eigenvalues are $0$ and $1$, respectively.
  We denote by $a$ and $b$ the two other eigenvalues such that $0 \le a \le b \le 1$ and we have that
  \begin{equation}
    A(a,b) = \ket{\psi_1}\ket{\psi_1} + b \ket{\psi_b}\bra{\psi_b} + a \ket{\psi_a}\bra{\psi_a} + 0 \ket{\psi_0}\bra{\psi_0}.
  \end{equation}
  Now, let $\rho$ be the quantum state
  \begin{equation}
    \rho = \lambda_1 \ket{\psi_1}\ket{\psi_1} + \lambda_2 \ket{\psi_b}\bra{\psi_b} + \lambda_3 \ket{\psi_a}\bra{\psi_a} + \lambda_4 \ket{\psi_0}\bra{\psi_0},
  \end{equation}
  where $\lambda_1 \ge \lambda_2 \ge \lambda_3 \ge \lambda_4 = 1 - \lambda_1 - \lambda_2 - \lambda_3 \ge 0$.
  This state is absolutely separable if and only if $(\lambda_1 - \lambda_3)^2 \le 4 \lambda_2 \lambda_4$ \cite{verstraete2001}.
  More specifically, we consider
  \begin{align}
    \lambda_1(a,b) &=  \frac{1}{4} \left( \frac{8 + \delta(a,b)}{\gamma(a,b)} - 1 \right), \\
    \lambda_2(a,b) = \lambda_3(a,b) &= \frac{1}{4} \left( 1 - \frac{\delta(a,b)}{\gamma(a,b)} \right), \\
    \lambda_4(a,b) = 1 - \lambda_1(a,b) - \lambda_2(a,b) - \lambda_3(a,b) &=  \frac{1}{4} \left( 3 - \frac{8 - \delta(a,b)}{\gamma(a,b)} \right), \\
  \end{align}
  where $\delta(a,b) = 1 - a - b$ and $\gamma(a,b) = \sqrt{8 + \delta^2(a,b)}$.
  From $0 \le a,b \le 1$ and $\gamma \le 3$, it follows that, indeed, $\lambda_1 \ge \lambda_2 = \lambda_3 \ge \lambda_4$.
  Furthermore, $\lambda_4 \ge 0$ is equivalent to $3 \gamma \ge 8 - \delta$ which holds due to $8 - \delta \ge 0$ and $(3\gamma)^2 - (8 - \delta)^2 = 8 (\delta + 1)^2 \ge 0$.
  Finally, we have that $(\lambda_1 - \lambda_3)^2 = 4 \lambda_2 \lambda_4$ and hence, $\rho(a,b)$ is an absolutely separable state for any $a,b$.
  Additionally, let us consider the state $\sigma$ with
  \begin{equation}
    \sigma = \nu_4 \ket{\psi_1}\ket{\psi_1} + \nu_3 \ket{\psi_b}\bra{\psi_b} + \nu_2 \ket{\psi_a}\bra{\psi_a} + \nu_1 \ket{\psi_0}\bra{\psi_0},
  \end{equation}
  where $\nu_j = \lambda_j(1-b,1-a)$.
  Since $\rho$ is an absolutely separable state for any $a,b$, so is $\sigma$.
  For the distance between their respective expectation values when measuring the observable $A$, we obtain
  \begin{equation}
    \trace A\rho - \trace A\sigma = - 1 + \frac{2}{\gamma} \left( 2 + \delta^2 \right),
  \end{equation}
  whose minimum $\sqrt{2} - 1$ is obtained at $\delta = 0$ and implies the bound $\mu_{2,1} \ge \sqrt{2} - 1$.
  
  To prove that this is the best we can do with absolutely separable states, we consider $A(\frac{1}{2},\frac{1}{2})$.
  It is well known that for a state $\rho$, there exists a von Neumann measurement with probability vector $p_j = \braket{a_j|\rho|a_j}$ for some basis $\ket{a_j}$ if, and only if, $\bm{p} \prec \bm{\lambda}$, i.e., the probability vector is majorized by the eigenvalue vector \cite{bengtsson2017}.
  A vector $\bm{b}$ is termed to majorize another vector $\bm{a}$, denoted by $\bm{b} \succ \bm{a}$, if $\sum_{j=1}^l b_j^\downarrow \ge \sum_{j=1}^l a_j^\downarrow$ for $l = 1, \dots, D$ with $D$ being the length of the vectors, where $\bm{b}^\downarrow$ and $\bm{a}^\downarrow$ are the same vectors as $\bm{b}$ and $\bm{a}$ but with components sorted in descending order, respectively.
    Also, note that for absolutely separable states, the eigenvectors of $A$ are irrelevant.
  An easy way to see this is that if an absolutely separable state $\rho$ realizes a certain probability vector $p_j = \braket{a_j|\rho|a_j}$, then the likewise absolutely separable state $U \rho U^\dagger$ realizes the same probability vector for the basis $\ket{\psi_j} = U \ket{a_j}$.
  This is possible because the separability only depends on the spectrum of the density matrix, not on the eigenvectors.
  Thus, the maximal expectation value with $A(\frac{1}{2},\frac{1}{2})$ for an absolutely separable state is given by the following optimization
  \begin{equation}
    \begin{aligned}
      \,&\maxover[\bm{p},\bm{\lambda}]  && p_1 + \frac{p_2+p_3}{2} \\
          &\subto && p_1 \le \lambda_1, \quad p_1 + p_2 \le \lambda_1 + \lambda_2, \quad p_1 + p_2 + p_3 \le \lambda_1 + \lambda_2 + \lambda_3, \\
          &       && p_1 \ge p_2 \ge p_3 \ge 1 - p_1 - p_2 - p_3 \ge 0, \\
          &       && \lambda_1 \ge \lambda_2 \ge \lambda_3 \ge 1 - \lambda_1 - \lambda_2 - \lambda_3 \ge 0, \\
          &       && (\lambda_1 - \lambda_3)^2 \le 4 \lambda_2 (1 - \lambda_1 - \lambda_2 - \lambda_3). \\
    \end{aligned}
  \end{equation}
  Clearly, for a given feasible point, we can increase $p_3$ while decreasing $p_2$ such that $p_2 + p_3 = \const$ since it leaves the objective value invariant and only relaxes the second constraint.
  Thus, there exists an optimal solution with $p_2 = p_3$.
  Now increasing $p_1$ while decreasing $p_2 = p_3$ such that $p_1 + p_2 = \const.$ leads to a relaxation of the constraints, and hence, we can require the optimum to satisfy $p_1 = \lambda_1$.
  In turn, we are left with maximizing $p_2$ depending on the $\lambda_j$ which gives $p_2 = p_3 = \frac{\lambda_2 + \lambda_3}{2}$ and leads to the simplified optimization
  \begin{equation}
    \begin{aligned}
      \,&\maxover[\bm{\lambda}]  && \lambda_1 + \frac{\lambda_2+\lambda_3}{2} \\
          &\subto && \lambda_1 \ge \lambda_2 \ge \lambda_3 \ge 1 - \lambda_1 - \lambda_2 - \lambda_3 \ge 0, \\
          &       && (\lambda_1 - \lambda_3)^2 \le 4 \lambda_2 (1 - \lambda_1 - \lambda_2 - \lambda_3). \\
    \end{aligned}
  \end{equation}
  It is certainly optimal to maximize $\lambda_1$ until either $1 - \lambda_1 - \lambda_2 - \lambda_3 = 0$ or $(\lambda_1 - \lambda_3)^2 = 4 \lambda_2 (1 - \lambda_1 - \lambda_2 - \lambda_3)$.
  In the first case, we can replace $\lambda_3$ by $\lambda_3 = 1 - \lambda_1 - \lambda_2$.
  Hence, we obtain the constraint $(2\lambda_1 + \lambda_2 -1)^2 \le 0$ implying $\lambda_2 = 1 - 2\lambda_1$.
  Then, it must hold that $\lambda_1 \ge 1 - 2\lambda_1 \ge \lambda_1$ and hence, $\lambda_1 = \frac{1}{3}$ leading to an optimal value of $2/3$.
  
  In the other case, we use the method of Lagrange multipliers for the added equality constraint and obtain $\lambda_1 = \frac{1}{2}$, $\lambda_2 = 0$, and $\lambda_3 = \frac{1}{2}$ which does not satisfy the inequalities.
  Hence, we look for the optimal solution at one of the boundaries $\lambda_1 = \lambda_2$, $\lambda_2 = \lambda_3$ or $\lambda_3 = 1 - \lambda_1 - \lambda_2 - \lambda_3$.
  
  For $\lambda_1 = \lambda_2$, we obtain via the method of Lagrange multipliers a single feasible point $\lambda_1 = \lambda_2 = (3 + \sqrt{3})/12$, $\lambda_3 = -1 + (3 + \sqrt{3})/4$ with objective value $(1 + \sqrt{3})/4 \approx 0.68$.
  In the case of $\lambda_2 = \lambda_3$, the same analysis gives the feasible point $\lambda_1 = (-1 + 2\sqrt{2})/4$, $\lambda_2 = \lambda_3 = 1/4$ with objective value $1/\sqrt{2} \approx 0.71$.
  The last case, $\lambda_3 = 1 - \lambda_1 - \lambda_2 - \lambda_3$, leads to an objective value of $(1 + \sqrt{3})/4 \approx 0.68$ at $\lambda_1 = (-3 + 5\sqrt{3})/12$, $\lambda_2 = (3 - \sqrt{3})/4$, and $\lambda_3 = (3 - \sqrt{3})/12$.
  
  Thus, the optimum is $1/\sqrt{2}$.\\
  
  For the minimal expectation value, we use the fact that
  \begin{equation}
    \min_\rho \trace A\rho = 1 - \max_\rho \trace (\mathds{1} - A)\rho.
  \end{equation}
  Since $\mathds{1} - A(\frac{1}{2},\frac{1}{2})$ is equivalent to $A(\frac{1}{2},\frac{1}{2})$ when we only consider eigenvalues, the minimal expectation value is given by $1 - 1/\sqrt{2}$ and the volume between maximal and minimal expectation value for absolutely separable states is $1/\sqrt{2} - (1 - 1/\sqrt{2}) = \sqrt{2} - 1$.
\end{proof}

\section{Volume calculation for multiple instances of two-qubit observables}\label{app:2quintances}
First, we show the calculation of the volume ratio for measurements $A_1 = 0 \oplus X \oplus 0$, $A_2 = 0 \oplus Y \oplus 0$, and $A_3 = 0 \oplus Z \oplus 0$, where $X$, $Y$, and $Z$ are the Pauli matrices.
Furthermore, we also consider the restriction to measurements $A_1$ and $A_2$ only.
The joint numerical range is given by a Bloch ball on the subspace spanned by $\ket{01}$ and $\ket{10}$, i.e., a ball of radius $1$.
Thus, the two- and three-dimensional volumes are given by $\pi$ and $\frac{4}{3}\pi$, respectively.
Second, the separable numerical range is symmetric w.r.t.~rotations around the axis of the third measurement.
More precisely, the local unitaries $U = U_1 \otimes U_2$, where
\begin{align}
  U_1 = \begin{pmatrix} 1 & 0 \\ 0 & e^{i\varphi/2} \end{pmatrix}, \quad
  U_2 = \begin{pmatrix} 1 & 0 \\ 0 & e^{-i\varphi/2} \end{pmatrix},
\end{align}
leave $A_3 = U A_3 U^\dagger$ invariant while continuously transforming $A_1$ and $A_2$ as $U A_1 U^\dagger = \cos\varphi A_1 + \sin\varphi A_2$ and $U A_2 U^\dagger = -\sin\varphi A_1 + \cos\varphi A_2$, respectively.
Because for local unitaries $U$, $\trace \rho_\text{sep} U A_i U^\dagger = \trace U^\dagger \rho_\text{sep} U A_i = \trace \sigma_\text{sep} A_i$ for any separable state $\rho_\text{sep}$ and separable $\sigma_\text{sep} = U^\dagger \rho_\text{sep} U$, the above transformation implies the rotational symmetry.
Thus, it is sufficient to solve the parametric optimization
\begin{equation}
  \begin{aligned}
        \,&\maxover[\rho \in \Sep]  && \trace \rho A_1 \\
          &\subto && \trace \rho A_2 = 0, \\
          &       && \trace \rho A_3 = c,
  \end{aligned}
\end{equation}
where $-1 \le c \le 1$.
As the separable numerical range is the convex hull of the pure-product numerical range, coming from pure product states, we consider general product states $\ket{\alpha}\ket{\beta}$ with
\begin{align}
\ket{\alpha} = \cos\frac{\alpha}{2} \ket{0} + e^{i\phi} \sin\frac{\alpha}{2} \ket{1}, \\
\ket{\beta} = \cos\frac{\beta}{2} \ket{0} + e^{i\psi} \sin\frac{\beta}{2} \ket{1}.
\end{align}
We have that $\bra{\alpha\beta} A_1 \ket{\alpha\beta} = \frac{1}{2} \cos(\phi - \psi) \sin\alpha \sin\beta$ and $\bra{\alpha\beta} A_3 \ket{\alpha\beta} = \frac{1}{2} (\cos\alpha - \cos\beta)$.
Hence, to maximize $\braket{A_1}$, certainly $\cos(\phi - \psi) = 1$ since we can choose the signs of $\sin\alpha$ and $\sin\beta$ independently from those of $\cos\alpha$ and $\cos\beta$.
The choice $\phi = \psi = 0$ not only gives $\cos(\phi - \psi) = 1$, but also makes sure $\ket{\alpha\beta}$ satisfies $\braket{A_2} = 0$ as $A_2$ is a skew-symmetric matrix and, in this case, $\ket{\alpha\beta}$ is a real-valued vector in the computational basis.
Actually, it is clear from the rotational symmetry that minimal and maximal $\langle A_1 \rangle$ are reached for $\langle A_2 \rangle = 0$.
Thus, we are left with optimizing
\begin{equation}
  \begin{aligned}
        \,&\maxover[\alpha,\beta]  && \sin\alpha \sin\beta \\
          &\subto && \cos\alpha - \cos\beta = c'.
  \end{aligned}
\end{equation}
To solve this, we write $x = \cos\alpha$, $y = -\cos\alpha$, and $\sin\alpha \sin\beta = \sqrt{(1-x^2)(1-y^2)}$ since we can always choose the positive solutions for the sines for given cosines.
Because the square root is a monotonic function, it is equivalent to maximize $(1-x^2)(1-y^2) = 1 - c'^2 + xy(xy+2)$ and, as $-1 \le xy$, also equivalent to maximize $xy$, leaving us with
\begin{equation}
  \begin{aligned}
        \,&\maxover[x,y]  && xy \\
          &\subto && x + y = c'.
  \end{aligned}
\end{equation}
By changing both the signs of $x$ and $y$, $c'$ can always be chosen nonnegative, and we have to find the rectangle with largest volume for given circumference, which is known to be a square.
Hence, $x = y$, or equivalently $\alpha = \beta$, provides the maximum with $\braket{A_1} = \frac{1}{2} (1 - c^2)$.
As this line provides the boundary of a convex set, the pure-product and separable numerical range coincide; see Fig.~\ref{fig:223_1over5} for a visualization.
The corresponding volume is given by
\begin{align}
  \vol L_\Sep(A_1,A_2,A_3) = \pi \int_{-1}^1 dc \left[ \frac{1}{2} (1 - c^2) \right]^2 = \frac{4}{15} \pi,
\end{align}
and the relative volume is $\frac{1}{5}$.
In the two-dimensional case restricted to observables $A_1$ and $A_2$, maximal $\bra{\alpha\beta}A_1\ket{\alpha\beta}$ independent from $c$ gives us the relevant volume as
\begin{align}
  \vol L_\Sep(A_1,A_2) = \left(\frac{1}{2}\right)^2 \pi = \frac{\pi}{4},
\end{align}
which leads to a relative volume of $\frac{1}{4}$.

\begin{figure}[t!]
  \centering
  \includegraphics[width=0.4\linewidth]{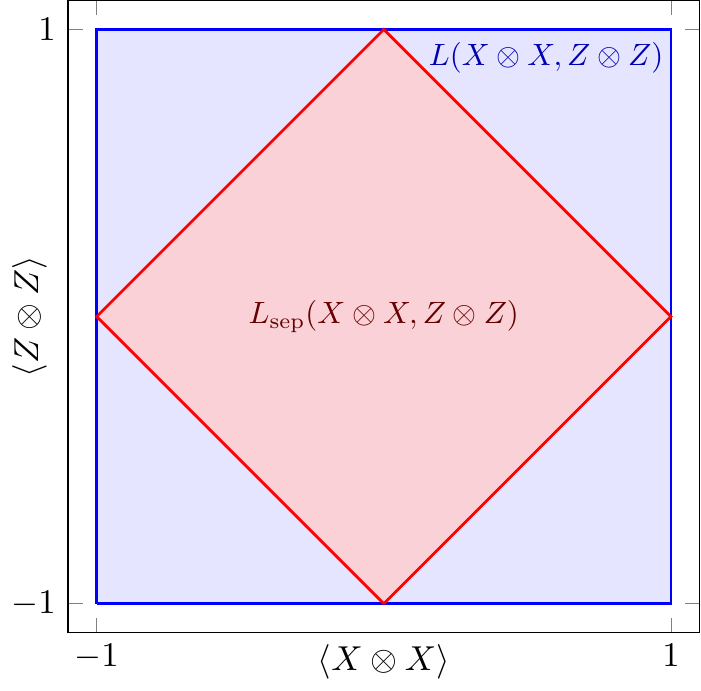}
  \caption{This figure shows the (separable) numerical range for observables $X \otimes X$ and $Z \otimes Z$.}
  \label{fig:sqinsq}
\end{figure}
Second, we consider the measurements $A_1 = X\otimes X$, $A_2 = X\otimes Y$, and $A_3 = Z\otimes Z$ which yield a ratio of $1/3$. 
For the computation, we again make use of the fact that simultaneous local unitary transformations of the observables do not alter the numerical range.
In this case, we consider $U = \mathds{1} \otimes \exp(-iZ\varphi/2)$ which leaves $A_3$ invariant and transforms $A_1$ and $A_2$ as
\begin{align}
  A_1(\varphi) &= \cos\varphi X\otimes X + \sin\varphi X\otimes Y, \\
  A_2(\varphi) &= -\sin\varphi X\otimes X + \cos\varphi X\otimes Y.
\end{align}
Thus, it corresponds to a rotation around the $A_3$-coordinate axis and hence, the numerical range is rotationally invariant around this axis and it suffices to consider $A_1$ and $A_3$.
The well known entanglement witnesses $W = \mathds{1} \pm X\otimes X \pm Z\otimes Z$ bound the separable numerical range, while with general quantum states any $\langle A_1 \rangle$ can be reached independently of $\langle A_3 \rangle$ and vice verse.
This leaves us with a square inside a square as shown in Fig.~\ref{fig:sqinsq}.
The volume ratio of the two solids of revolution generated by rotating around the $A_3$-axis is $1/3$, given by two cones glued to each other at the base inside a cylinder.

Finally, additionally measuring the observable $A_4 = Y\otimes Z$ yields a volume ratio of $1/6$.
Similar to before, we, in addition, consider the local unitary transformation $V = \exp(-iX\phi/2) \otimes \mathds{1}$, which leaves $A_1$ and $A_2$ invariant, while transforming $A_3$ and $A_4$ as
\begin{align}
  A_3(\phi) &= \cos\phi Z\otimes Z + \sin\phi Y\otimes Z, \\
  A_4(\phi) &= -\sin\phi Z\otimes Z + \cos\phi Y\otimes Z.
\end{align}
Thus, the resulting four-dimensional solids of revolution have a volume of
\begin{align}
  \vol L &= \int_0^1 dR \int_0^{2\pi} d\varphi R \int_0^1 dr \int_0^{2\pi} r = \pi^2, \\
  \vol L_\Sep &= \int_0^1 dR \int_0^{2\pi} d\varphi R \int_0^{1-R} dr \int_0^{2\pi} r = \frac{\pi^2}{6},
\end{align}
and hence, we obtain a volume ratio of $1/6$.

\section{Proof of Theorem~\ref{thm:222sep}}\label{app:222sep}
In this Appendix, the following result is established.
\setcounter{definition}{8}
\begin{theorem}
For locally traceless two-qubit product observables $A_1$ and $A_2$ written in the standard form as in Eqs.~(\ref{eq:ltobs1},\,\ref{eq:ltobs2}), the volume ratio is 
\begin{equation}
    \frac{\vol L_\Sep}{\vol L} = \frac{\pi}{8} \left[ \left( \left|\sin\theta_-\right|+\left|\sin\theta_+\right| \right) - \frac{ \tilde{F}(\theta_-,\theta_+)}{\tilde{T}(\theta_-,\theta_+)} \right] /
                \left[ \cos\theta_- - \cos\theta_+ + G_-(\theta_-) + G_+(\theta_+) \right],
\end{equation}
where $\theta_\pm = \theta_A \pm \theta_B$ as well as $G_-(\theta_-) = |\frac{\theta_-}{2} \sin\theta_-|$, $G_+(\theta_+) = |\left( \frac{\theta_+}{2}-\frac{\pi}{2} \right) \sin\theta_+|$, $\tilde{F}(\theta_-,\theta_+) = F(\theta_-,\theta_+) - F(\theta_+,\theta_-)$, $\tilde{T}(\theta_-,\theta_+) = T(\theta_-,\theta_+) - T(\theta_+,\theta_-)$, and the functions $F$ and $T$ are given by
\begin{align}
  F(x,y) &= \left| \sin\frac{x}{2}\cos\frac{y}{2} \right| \left[ K\left( 1-T^2(x,y) \right) - E\left( 1-T^2(x,y) \right) \right], \\
  T(x,y) &= \left| \frac{\tan\frac{x}{2}}{\tan\frac{y}{2}} \right|,
\end{align}
where $K(\cdot)$ and $E(\cdot)$ are the elliptic integrals of the first and second kind, respectively.
\end{theorem}
Since rescaling as well as local unitary transformations preserve the volume ratio, it is sufficient to consider observables $A_1 = X \otimes X$ and $A_2 = (\cos\theta_A X + \sin\theta_A Z) \otimes (\cos\theta_B X + \sin\theta_B Z)$ with $0 \le \theta_A, \theta_B \le \pi$.
To prove Theorem~\ref{thm:222sep}, we compute the separable as well as the standard numerical range for these observables explicitly via the following Lemmata.
\setcounter{definition}{11}
\begin{lemma}\label{lem:twoqu_seploctr}
For two-qubit observables $A_1 = X \otimes X$ and $A_2 = (\cos\theta_A X + \sin\theta_A Z) \otimes (\cos\theta_B X + \sin\theta_B Z)$ where $0 \le \theta_A, \theta_B \le \pi$, the volume of the separable numerical range is
\begin{equation}
    \vol L_\Sep = \frac{\pi}{4}\left( \left|\sin\theta_-\right|+\left|\sin\theta_+\right| \right) + \frac{2 \left[ F(\theta_-,\theta_+) - F(\theta_+,\theta_-) \right]}{T(\theta_+,\theta_-) - T(\theta_-,\theta_+)},
\end{equation}
where the functions $F$ and $T$ are given by
\begin{align}
  F(x,y) &= \left| \sin\frac{x}{2}\cos\frac{y}{2} \right| \left[ K\left( 1-T^2(x,y) \right) - E\left( 1-T^2(x,y) \right) \right], \\
  T(x,y) &= \left| \frac{\tan\frac{x}{2}}{\tan\frac{y}{2}} \right|
\end{align}
and $\theta_\pm = \theta_A \pm \theta_B$ and $K(\cdot)$ and $E(\cdot)$ are the elliptic integrals of the first and second kind, respectively.
\end{lemma}
\begin{proof}
  For a general product state
  \begin{equation}
     \ket{\psi_{\otimes}} = \left( \cos\theta_1 \ket{0} + e^{i\phi_1}\sin\theta_1 \ket{1} \right) \otimes \left( \cos\theta_2 \ket{0} + e^{i\phi_2}\sin\theta_2 \ket{1} \right),
  \end{equation}
  we have that 
  \begin{align}
    \bra{\psi_{\otimes}} A_1 \ket{\psi_{\otimes}} = &\cos\theta_1 \cos\theta_2, \\
    \bra{\psi_{\otimes}} A_2 \ket{\psi_{\otimes}} = &\left( \cos\theta_1 \cos\theta_A + \cos\phi_1 \sin\theta_1 \sin\theta_A \right) \times \left( \cos\theta_2 \cos\theta_B + \cos\phi_2 \sin\theta_2 \sin\theta_B \right).
  \end{align}
  Clearly, for fixed $\bra{\psi_{\otimes}} A_1 \ket{\psi_{\otimes}}$, the extremal points of $\bra{\psi_{\otimes}} A_2 \ket{\psi_{\otimes}}$ are reached for $\cos\phi_1 = \pm 1$ and $\cos\phi_2 = \pm 1$.
  Thus, allowing $\theta_1,\theta_2 \in [0,2\pi)$, we can fix $\cos\phi_1 = \cos\phi_2 = 1$ to calculate the boundary of the separable numerical range.
  Hence, we obtain
  \begin{align}
    \bra{\psi_{\otimes}} A_1 \ket{\psi_{\otimes}} &= \frac12 \left( \cos\theta_- + \cos\theta_+ \right), \\
    \bra{\psi_{\otimes}} A_2 \ket{\psi_{\otimes}} &= \frac12 \left[ \cos(\theta_- - \theta_-) + \cos(\theta_+ - \theta_+) \right],
  \end{align}
  where $\theta_\pm = \theta_1 \pm \theta_2$ and $\theta_\pm = \theta_A \pm \theta_B$.
  To simplify the calculation, it is advantageous to perform a rotation by $-\frac{\pi}{4}$ such that $\tilde{A}_1 = \frac{1}{\sqrt{2}} (A_1 + A_2)$ and $\tilde{A}_2 = \frac{1}{\sqrt{2}} (-A_1 + A_2)$.
  Again, this does not change the volume of the separable numerical range.
  Then,
  \begin{align}
    \bra{\psi_{\otimes}} \tilde{A}_1 \ket{\psi_{\otimes}} &= \frac{1}{\sqrt{2}} \left( \cos\frac{\theta_-}{2} \cos z_- + \cos\frac{\theta_+}{2} \cos z_+ \right), \\
    \bra{\psi_{\otimes}} \tilde{A}_2 \ket{\psi_{\otimes}} &= \frac{1}{\sqrt{2}} \left( \sin\frac{\theta_-}{2} \sin z_- + \sin\frac{\theta_+}{2} \sin z_+ \right),
  \end{align}
  where $z_\pm = \frac{\theta_\pm}{2} - \theta_\pm \in [0,2\pi)$.
  Thus, the enclosed area is the Minkowski sum of two ellipses, and hence, it is convex.
  Indeed, in this case, the product numerical range and the separable numerical range coincide as varying $z_\pm$ continuously traces out the entire boundary.
  Therefore, to calculate the volume, we fix $\bra{\psi_{\otimes}} \tilde{A}_2 \ket{\psi_{\otimes}}$ while maximizing and minimizing $\bra{\psi_{\otimes}} \tilde{A}_1 \ket{\psi_{\otimes}}$.
  Because we can change the signs of $\cos z_\pm$ and $\sin z_\pm$ independently, the maximum and minimum actually coincide apart from the sign.
  Thus, it is sufficient to solve the optimization problem
  \begin{equation}\label{eq:prod_opt}
    \begin{aligned}
      \max_{z_\pm \in [0,\frac{\pi}{2}]} &\left| \cos\frac{\theta_-}{2} \right| \cos z_- + \left| \cos\frac{\theta_+}{2} \right| \cos z_+ \\
                           \text{s.t. } &\left| \sin\frac{\theta_-}{2} \right| \sin z_- + \left| \sin\frac{\theta_+}{2} \right| \sin z_+ = |c|,
    \end{aligned}
  \end{equation}
  where $|c| \in [0,\left| \sin\frac{\theta_-}{2} \right| + \left| \sin\frac{\theta_+}{2} \right|$.
  We can restrict ourselves to $z_\pm \in [0,\frac{\pi}{2}]$ since it is obviously optimal to have $\cos z_\pm \ge 0$, and if one of the $\sin z_\pm$ were negative, making it positive would reduce the other and lead to a larger objective value.
  Since
  \begin{equation}
  \frac{d\left( \left| \cos\frac{\theta_\pm}{2} \right| \cos z_\pm \right)}{d\left( \left| \sin\frac{\theta_\pm}{2} \right| \sin z_\pm \right)} = - \frac{\tan z_\pm}{\left| \tan\frac{\theta_\pm}{2} \right|},
  \end{equation}
  a maximum is reached when $\frac{\tan z_-}{\left| \tan\frac{\theta_-}{2} \right|} = \frac{\tan z_+}{\left| \tan\frac{\theta_+}{2} \right|}$.
  This equation always describes a feasible point of the optimization as $\tan z_\pm$ approaches infinity when $\sin z_\pm$ goes to $1$, and hence, continuously changing $\sin z_-$ from its minimal allowed value for a given $|c|$ to $1$ leads to a continuous change of $\sin z_+$ from $1$ to its minimal allowed value via the constraint, and somewhere along the way, the condition $\frac{\tan z_-}{\left| \tan\frac{\theta_-}{2} \right|} = \frac{\tan z_+}{\left| \tan\frac{\theta_+}{2} \right|}$ is satisfied.
  Thus, we obtain
  \begin{equation}
    \vol L_\Sep(A_1,A_2) = 2 \int_0^{\left| \sin\frac{\theta_-}{2} \right| + \left| \sin\frac{\theta_+}{2} \right|} \text{opt}(c) dc,
  \end{equation}
  where $\text{opt}(c)$ is the result of the optimization in Eq.~(\ref{eq:prod_opt}).
  Since we have that $z_\pm \in [0,\frac{\pi}{2}]$,
  \begin{align}
    c &= \left| \sin\frac{\theta_-}{2} \right| \sin z_- + \left| \sin\frac{\theta_+}{2} \right| \sin z_+ \nonumber \\
      &= \left| \sin\frac{\theta_-}{2} \right| \frac{\tan z_-}{\sqrt{1 + \tan^2 z_-}} + \left| \sin\frac{\theta_+}{2} \right| T(\theta_+,\theta_-) \frac{\tan z_-}{\sqrt{1 + \left( T(\theta_+,\theta_-) \tan z_- \right)^2}}, \\
    \text{opt}(c) &= \left| \cos\frac{\theta_-}{2} \right| \cos z_- + \left| \cos\frac{\theta_+}{2} \right| \cos z_+ \nonumber \\
                  &= \left| \cos\frac{\theta_-}{2} \right| \frac{1}{\sqrt{1 + \tan^2 z_-}} + \left| \cos\frac{\theta_+}{2} \right| \frac{1}{\sqrt{1 + \left( T(\theta_+,\theta_-) \tan z_- \right)^2}},
  \end{align}
  changing the integration variable from $c$ to $u = \tan z_-$ leads to
  \begin{equation}
    \begin{aligned}
    \vol L_\Sep &= \int_0^{\infty} du \left[ \frac{\left| \sin \theta_- \right|}{\left( 1 + u^2 \right)^2} + \frac{\left| \sin\theta_+ \right| T(\theta_+,\theta_-)}{\left( 1 + T^2(\theta_+,\theta_-) u^2 \right)^2} + \frac{2 \left| \sin\frac{\theta_-}{2} \cos\frac{\theta_+}{2} \right|}{\left( 1 + u^2 \right)^{\frac{3}{2}} \left( 1 + T^2(\theta_+,\theta_-) u^2 \right)^{\frac{1}{2}}}\right. \\
                         &\hspace*{30pt}\left. + \frac{2 \left| \cos\frac{\theta_-}{2} \sin\frac{\theta_+}{2} \right| T(\theta_+,\theta_-)}{\left( 1 + T^2(\theta_+,\theta_-) u^2 \right)^{\frac32} \left( 1 + u^2 \right)^{\frac{1}{2}}} \right] \\
                         &= \frac{\pi}{4}\left( \left|\sin\theta_-\right|+\left|\sin\theta_+\right| \right) + \frac{2\left[ F(\theta_-,\theta_+) - F(\theta_+,\theta_-) \right]}{T(\theta_+,\theta_-) - T(\theta_-,\theta_+)} ,
  \end{aligned}
  \end{equation}
  proving the lemma.
\end{proof}
Fortunately, there is a known procedure for the computation of the joint numerical range of two Hermitian matrices, which we use in the following proof.
\setcounter{definition}{12}
\begin{lemma}\label{lem:twoqu_loctr}
For two-qubit observables $A_1 = X \otimes X$ and $A_2 = (\cos\theta_A X + \sin\theta_A Z) \otimes (\cos\theta_B X + \sin\theta_B Z)$ where $0 \le \theta_A, \theta_B \le \pi$, the volume of the joint numerical range is
\begin{equation}
  \vol L(A_1, A_2) = 2 \left[ \cos\theta_- - \cos\theta_+ + |\frac{\theta_-}{2} \sin\theta_-| + |\left( \frac{\theta_+}{2}-\frac{\pi}{2} \right) \sin\theta_+| \right]
\end{equation}
where $\theta_\pm = \theta_A \pm \theta_B$.
\end{lemma}
\begin{proof}
  We calculate the numerical range explicitly using its generating line $C(A_1,A_2)$ defined by the dual (line) equation
  \begin{equation}\label{eq:genline}
    \det(u A_1 + v A_2 + w \mathds{1}) = 0,
  \end{equation}
  where $u x + v y + w = 0$ is the equation of a supporting line to $L(A_1,A_2)$ in the $x$-$y$-plane.
  Then, the numerical range is given by the convex hull of its generating line \cite{murnaghan1932, kippenhahn1951}.
  We use a usual procedure as described in Ref.~\cite{fiedler1981}, i.e., we dehomogenize by setting $v = 1$, replace $w$ in Eq.~(\ref{eq:genline}) by $w = -u x - y$, and solve for the generating line with the resulting equations
  \begin{equation}
    F(u,x,y) = \det(u A_1 + A_2 - (u x + y) \mathds{1}) = 0,
  \end{equation}
  as well as $\partial F(u,x,y) / \partial u = 0$.
  We compute a Gr{\"o}bner basis \cite{buchberger1970} of this system of polynomial equations that contains the polynomial
  \begin{equation}
  \begin{aligned}
    P(x,y) &= \frac{1}{8} \left[ 2 (-1 + 2 x^2 + 2 y^2)^2 + \cos(4 \theta_A) + 4 \cos(2 \theta_A) (4 x y (x y - \cos\theta_A \cos\theta_B) + (-1 + 2 x^2 + 2 y^2) \cos(2 \theta_B) \right. \\
    &\left. + 16 x y (x y \cos(2 \theta_B) - \cos\theta_A \cos\theta_B (2 (-1 + x^2 + y^2) + \cos(2 \theta_B))) + \cos(4 \theta_B) \right] (-1 + y^2)^2 \sin^2\theta_A \sin^2\theta_B.
  \end{aligned}  
  \end{equation}
  Hence, $P(x,y) = 0$ has to be satisfied for a solution of the system of polynomial equations determining the generating line.
  
  The special cases $\sin\theta_A = 0$ and $\sin\theta_B = 0$ work analogously.
  Thus, let $\sin\theta_A = 0$ and hence, $A_2 = X \otimes (\cos\theta_B X + \sin\theta_B Z)$ which means that there is only a single measurement on the first particle.
  This can be simulated by a separable state in place of a possibly entangled state $\ket{\psi}$.
  Indeed, the state
  \begin{equation}
    \rho_{\text{sep}} = \rho_A \otimes \frac{ \trace_A[ (X \otimes \mathds{1})\ket{\psi}\bra{\psi} ] }{ \trace[ (X \otimes \mathds{1})\ket{\psi}\bra{\psi} ] },
  \end{equation}
  with $\trace X\rho_A = \trace[ (X \otimes \mathds{1})\ket{\psi}\bra{\psi} ]$, yields the same $X\otimes M$ measurement statistics for any observable $M$.
  Hence, it is sufficient to consider product states $\rho = \rho_A \otimes \rho_B$.
  For fixed $\langle A_1 \rangle = \trace X\rho_A \trace X\rho_B$, we optimize $\langle A_2 \rangle = \trace X\rho_A \trace(\cos\theta_B X + \sin\theta_B Z)\rho_B$ using the notation $r_X = \trace X\rho_A$, $s_X = \trace X\rho_B$, and $s_Z = \trace Z\rho_B$,
  \begin{equation}
  \begin{aligned}
    \optover[r_X,s_X,s_Z] &r_X (\cos\theta_B s_X + \sin\theta_B s_Z) \\
    \subto\quad &r_X s_X = \langle A_1 \rangle, \\
           &-1 \le r_X \le 1, \\
           &s_X^2 + s_Z^2 \le 1,
  \end{aligned}
  \end{equation}
  to get the boundary of the product numerical range whose convex hull is the joint numerical range.
  Replacing $r_X$ in the optimization by $\langle A_1 \rangle / s_X$ leads to
  \begin{equation}
  \begin{aligned}
    \optover[r_X,s_X,s_Z] &\langle A_1 \rangle (\cos\theta_B + \sin\theta_B \frac{s_Z}{s_X}) \\
    \subto\quad &|\langle A_1 \rangle| \le |s_X|, \\
           &s_X^2 + s_Z^2 \le 1.
  \end{aligned}
  \end{equation}
  The optimum is given by $\langle A_1 \rangle \cos\theta_B \pm \sqrt{1-\langle A_1 \rangle^2}\sin\theta_B$ with $|r_X|=1$, $|s_X| = |\langle A_1 \rangle|$, and $s_Z = \sqrt{1-s_X^2}$.
  Fortunately, the maximum and the minimum are concave and convex functions in $\langle A_1 \rangle$, respectively, implying that they also give the boundary of the numerical range.
  Then, the volume is
  \begin{equation}
    \vol L(X\otimes X, X\otimes (\cos\theta_B X + \sin\theta_B Z)) = \int_{-1}^1 d\xi ~2 \sin\theta_B \sqrt{1-\xi^2} = \pi \sin\theta_B.
  \end{equation}
  Indeed, since the numerical range, and hence also its volume, changes continuously with the parameters $\theta_A$ and $\theta_B$, this is the limit of the general volume function.
  
  The other special case, i.e.~$(1-y^2)=0$, describes the eigenspace with maximal or minimal eigenvalue of $A_2$.
  Because $A_1$ and $A_2$ behave the same relative to each other, we consider states $\ket{\psi}$ with $\bra{\psi_\pm} A_1 \ket{\psi_\pm} = \pm 1$, which means that $\ket{\psi_+} = \alpha \ket{++} + \beta \ket{--}$ and $\ket{\psi_-} = \alpha \ket{+-} + \beta \ket{-+}$.
  Then, we obtain
  \begin{align}
    \bra{\psi_+} A_2 \ket{\psi_+} &= \cos\theta_A \cos\theta_B + \sin\theta_A \sin\theta_B (\alpha\beta^* + \alpha^* \beta), \\
    \bra{\psi_-} A_2 \ket{\psi_-} &= - \cos\theta_A \cos\theta_B + \sin\theta_A \sin\theta_B (\alpha\beta^* + \alpha^* \beta),
  \end{align}
  and hence, this solution yields line segments with $\pm \cos\theta_A \cos\theta_B - \sin\theta_A \sin\theta_B \le \langle A_1 \rangle \le \pm \cos\theta_A \cos\theta_B + \sin\theta_A \sin\theta_B$ for $\langle A_2 \rangle = \pm 1$, respectively.
  Let us keep this special case in mind for the final solution.
  
  Finally, we consider the case
  \begin{equation}
  \begin{aligned}
    \tilde{P}(x,y) = &2 (-1 + 2 x^2 + 2 y^2)^2 + \cos(4 \theta_A) + 4 \cos(2 \theta_A) (4 x y (x y - \cos\theta_A \cos\theta_B) + (-1 + 2 x^2 + 2 y^2) \cos(2 \theta_B) \\
                     &+ 16 x y (x y \cos(2 \theta_B) - \cos\theta_A \cos\theta_B (2 (-1 + x^2 + y^2) + \cos(2 \theta_B))) + \cos(4 \theta_B) = 0.
  \end{aligned}  
  \end{equation}
  Rotating the coordinate system by $\frac{\pi}{4}$, i.e.~setting $x = (\tilde{x}+\tilde{y})/\sqrt{2}$ and $y = (\tilde{x}-\tilde{y})/\sqrt{2}$, we find that the solution is given by two ellipses with semi-axes along the coordinate axes,
  \begin{equation}
  \begin{aligned}
    \tilde{P}(\tilde{x},\tilde{y}) = 8 &\left[ 1 - \cos^2(\theta_A+\theta_B) - \tilde{x}^2 (1 - \cos(\theta_A+\theta_B)) - \tilde{y}^2 (1 + \cos(\theta_A+\theta_B)) \right] \times \\
                                       &\left[ 1 - \cos^2(\theta_A-\theta_B) - \tilde{x}^2 (1 - \cos(\theta_A-\theta_B)) - \tilde{y}^2 (1 + \cos(\theta_A-\theta_B)) \right].
  \end{aligned}  
  \end{equation}
\begin{figure}[t!]
  \centering
  \includegraphics[width=.6\linewidth]{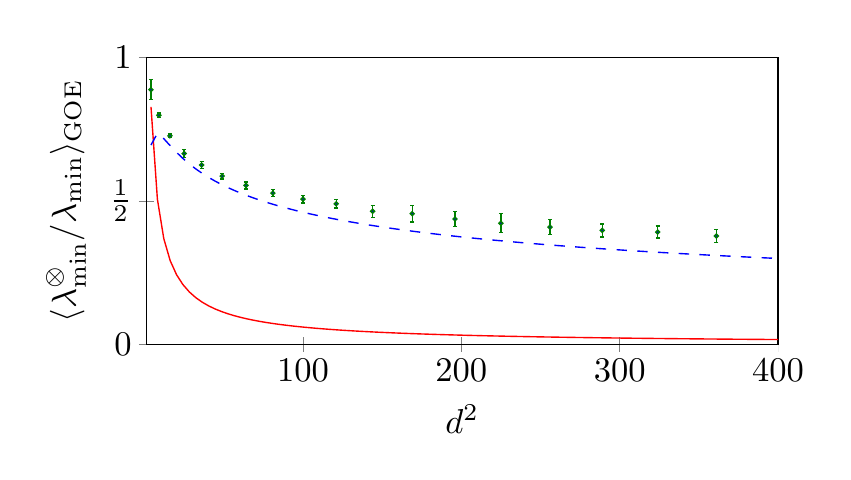}
  \vspace*{-0.5cm}
  \caption{Visual representation of the ratio $\left\langle {\lambda^\otimes_{\text{min}}}/{\lambda_{\text{min}}} \right\rangle_{\text{GOE}}$, see Table \ref{tab:goelambdaotimesmin}. The red solid curve is a lower bound of the ratio based on the diagonal entries ($2 \sqrt{2} d^{-1} \sqrt{\ln d}$) and the blue dashed line is an educated guess based on the multi-qubit \emph{upper} bound ($4 d^{-1/2} \ln d$) both derived in Ref.~\cite{czartowski2019}. The error bars are scaled by a factor of 10 to increase visibility.}\label{fig:goelambdaotimesmin}
\end{figure}
  To show that these curves are indeed reached by quantum states, we find those states explicitly.
  Rotating the coordinate system corresponds to considering the observables $A_\pm = (A_1 \pm A_2)/\sqrt{2}$.
  For the first ellipse with semi-axes of length $l_\pm^a = \sqrt{1 \pm \cos(\theta_A + \theta_B)}$, the states are given by $\ket{\phi_a} = \cos\varphi \ket{\phi_+^a}/\sqrt{\braket{\phi_+^a|\phi_+^a}} + \sin\varphi \ket{\phi_-^a}/\sqrt{\braket{\phi_-^a|\phi_-^a}}$, where
  \begin{align}
    \ket{\phi_+^a} &= \sin(\theta_A+\theta_B) (\ket{00} - \ket{11}) + \left[ 1 + \cos(\theta_A+\theta_B) + \sqrt{2 + 2\cos(\theta_A+\theta_B)}\right] (\ket{01} + \ket{10}), \\
    \ket{\phi_-^a} &= \sin(\theta_A+\theta_B) (\ket{00} - \ket{11}) - \left[ 1 - \cos(\theta_A+\theta_B) + \sqrt{2 - 2\cos(\theta_A+\theta_B)}\right] (\ket{01} + \ket{10}),
  \end{align}
  are eigenstates of $A_\pm$ with eigenvalues $l_\pm^a$, respectively.
  There is an exception when $\theta_A+\theta_B = \pi$ and hence, $\ket{\phi_+^a} = 0$.
  In this case, the ellipse is just a line segment of length $\sqrt{2}$ along the $y$-axis.
  Since the respective eigenvectors of $A_-$ with eigenvalues $\pm\sqrt{2}$ obey $\langle A_+ \rangle = 0$, the line is traced out by mixtures of these states.
  Otherwise, it holds that $\braket{\phi_+^a|\phi_-^a}/\sqrt{\braket{\phi_+^a|\phi_+^a}\braket{\phi_-^a|\phi_-^a}} = -1/\sqrt{2}$, and thus, $\braket{\phi_a|\phi_a} = 1-\sqrt{2}\sin\varphi \cos\varphi$, as well as $\bra{\phi_\pm^a} A_\mp \ket{\phi_\pm^a} = 0$ and hence, $\bra{\phi_a} A_+ \ket{\phi_a} = l_+^a (\cos^2\varphi - \sqrt{2}\sin\varphi \cos\varphi)$ and $\bra{\phi_a} A_- \ket{\phi_a} = l_-^a (\cos^2\varphi - \sqrt{2}\sin\varphi \cos\varphi)$.
  Thus, the states $\ket{\phi_a}$ satisfy
  \begin{align}
    \langle A_+ \rangle &= \frac{\bra{\phi_a} A_+ \ket{\phi_a}}{\braket{\phi_a|\phi_a}} = l_+^a \frac{\cot\varphi (\cot\varphi - \sqrt{2})}{1 + \cot\varphi (\cot\varphi - \sqrt{2})}, \\
    \langle A_- \rangle &= \frac{\bra{\phi_a} A_- \ket{\phi_a}}{\braket{\phi_a|\phi_a}} = l_-^a \frac{\tan\varphi (\tan\varphi - \sqrt{2})}{1 + \tan\varphi (\tan\varphi - \sqrt{2})},
  \end{align}
  which indeed defines points on the desired ellipse because $(\langle A_+ \rangle / l_+^a)^2 + (\langle A_- \rangle / l_-^a)^2 = 1$.
  Furthermore, with $\varphi$ varying continuously from $0$ to $\pi$, we observe the continuous variation of $(\langle A_+ \rangle, \langle A_- \rangle)$ as $(1,0) \rightarrow (0,-1) \rightarrow (-1,0) \rightarrow (0,1) \rightarrow (1,0)$ where the change in-between each step is monotonic.
  Thus, the states indeed trace out the ellipse.
  For the second ellipse, an analogous argument holds for the states $\ket{\phi_b} = \cos\varphi \ket{\phi_+^b}/\sqrt{\braket{\phi_+^b|\phi_+^b}} + \sin\varphi \ket{\phi_-^b}/\sqrt{\braket{\phi_-^b|\phi_-^b}}$, where
  \begin{align}
    \ket{\phi_+^b} &= \sin(\theta_A-\theta_B) (\ket{00} + \ket{11}) - \left[ 1 + \cos(\theta_A-\theta_B) - \sqrt{2 + 2\cos(\theta_A-\theta_B)}\right] (\ket{01} - \ket{10}), \\
    \ket{\phi_-^b} &= \sin(\theta_A-\theta_B) (\ket{00} + \ket{11}) + \left[ 1 - \cos(\theta_A-\theta_B) - \sqrt{2 - 2\cos(\theta_A-\theta_B)}\right] (\ket{01} - \ket{10}).
  \end{align}
  Obviously, the convex hull of these two ellipses is bounded partially by straight line segments.
  These line segments must indeed correspond to the maximal and minimal eigenvalues of $A_1$ and $A_2$ which we considered in the special case $y = \pm 1$ since there cannot be any state beyond $-1 \le x,y \le 1$.
  To compute the volume, we divide the convex hull of the ellipses into a polygon with eight vertices and the tips of the ellipses, whose volume can be easily computed via integration; see Fig.~\ref{fig:prod_obs}.
  This leads to the final result
  \begin{equation}
  \begin{aligned}
    \vol L(A_1, A_2) &= 2 \left[ \cos\theta_- - \cos\theta_+ + |\sin\theta_-| \arccos|\cos\frac{\theta_-}{2}| + |\sin\theta_+| \arccos|\sin\frac{\theta_+}{2}| \right] \\
                     &= 2 \left[ \cos\theta_- - \cos\theta_+ + |\frac{\theta_-}{2} \sin\theta_-| + |\left( \frac{\theta_+}{2}-\frac{\pi}{2} \right) \sin\theta_+| \right],
  \end{aligned}
  \end{equation}
  where $\theta_\pm = \theta_A \pm \theta_B$.
\end{proof}
 
\section{Minimal separable expectation value}\label{app:minsepexpval}
In this Appendix, numerical results concerning generic random observables of size $D = d^2$ are presented.
The obtained data are listed in Table~\ref{tab:goelambdaotimesmin} and visualized in Fig.~\ref{fig:goelambdaotimesmin}.

\renewcommand{\arraystretch}{1.4}
\begin{table}[h!]
	\centering
	\begin{tabular}{c|c|c||c|c|c}
$d$ & $\left\langle {\lambda^\otimes_{\text{min}}}/{\lambda_{\text{min}}} \right\rangle_{\text{GOE}}$ & sample size  & $d$ & $\left\langle {\lambda^\otimes_{\text{min}}}/{\lambda_{\text{min}}} \right\rangle_{\text{GOE}}$ & sample size\\[3pt] \hline
 2 & 0.8871(4) & 15 000 & 11 & 0.490(1) & 300 \\
 3 & 0.7981(8) &  15 000 & 12 & 0.464(1) & 125 \\
 4 & 0.7265(6) &  15 000 & 13 & 0.4554(3) & 40 \\
 5 & 0.665(1) & 2 000 & 14 & 0.4368(3) & 40 \\
 6 & 0.625(1) & 2 000 & 15 & 0.4222(3) & 40 \\
 7 & 0.586(1) & 2 000 & 16 & 0.4088(3) & 40 \\
 8 & 0.554(1) & 1 000 & 17 & 0.3974(2) & 40 \\
 9 & 0.527(1) & 750 & 18 & 0.3915(2) & 40 \\
 10 & 0.506(1) & 500 & 19 & 0.3778(2) & 40 \\\hline
	\end{tabular}
	\caption{Ratio of the minimal separable expectation value $\lambda^\otimes_{\min}$ to the minimal eigenvalue $\lambda_{\min}$, averaged over samples of random symmetric matrices $X$ from the Gaussian Orthogonal Ensemble. Results were obtained by numerical global optimalization of Eq.~\eqref{eqn:lotimesexpr}.}
	
	\label{tab:goelambdaotimesmin}
\end{table}
\renewcommand{\arraystretch}{1}

\twocolumngrid

\bibliography{Bibliography}

\end{document}